\documentclass[journal, letter]{IEEEtran} 

\usepackage{amsmath,epsfig, cite, cases}
\usepackage{subfigure,wrapfig}
\usepackage{amssymb,amscd,mathrsfs}
 \usepackage{indentfirst}
\usepackage[active]{srcltx}
\usepackage{color}
\usepackage{setspace}
\usepackage{enumerate}
\usepackage{bbding}
\usepackage{algorithm, algorithmic}
\usepackage{bbm}
\usepackage{mathtools}

\allowdisplaybreaks[1]
\newtheorem{thm}{Theorem}

\newtheorem{lem}{Lemma}
\newtheorem{prop}{Proposition}

\def\E{\mathbb{E}}

\title{Near Optimal Energy Control and Approximate Capacity of Energy Harvesting Communication}
\author{Yishun Dong, Farzan Farnia, and Ayfer \"Ozg\"ur,~\IEEEmembership{Member,~IEEE}
\thanks{Manuscript received April 1, 2014; revised September 1, 2014. This work was supported by the Center for Science of Information (CSoI), an NSF Science and Technology Center, under grant agreement CCF-0939370. An initial version of this paper was presented at the IEEE International Symposium on Information Theory (ISIT), Honolulu, July 2014.}
\thanks{The authors are with Stanford University, Packard Electrical Engineering, 350 Serra Mall, Stanford, California 94305-9510, USA (e-mails: \{ydong2, farnia, aozgur\}@ stanford.edu).}
}

\begin{document}

\maketitle
\thispagestyle{empty}

\begin{abstract}

We consider an energy-harvesting communication system where a transmitter powered by an exogenous energy arrival process and equipped with a finite battery of size $B_{max}$ communicates over a discrete-time AWGN channel. We first concentrate on a simple Bernoulli energy arrival process where at each time step, either an energy packet of size $E$ is harvested with probability $p$, or no energy is harvested at all, independent of the
other time steps. We provide a near optimal energy control policy and a simple approximation to the information-theoretic capacity of this channel. Our approximations for both problems are universal in all the system parameters involved ($p$, $E$ and $B_{max}$), i.e. we bound the approximation gaps by a constant independent of the parameter values. Our results suggest that a battery size $B_{max}\geq E$ is (approximately) sufficient to extract the infinite battery capacity of this channel. We then extend our results to general i.i.d. energy arrival processes. Our approximate capacity characterizations provide important insights for the optimal design of energy harvesting communication systems in the regime where both the battery size and the average energy arrival rate are large.    

\end{abstract}
\section{Introduction}
\label{sec: Introduction}

In many future wireless networks, we may encounter nodes (such as sensor nodes) that harvest the energy they need for communication from the natural resources in their environment. The simplest model that captures this communication scenario is the discrete-time AWGN channel depicted in Fig.~\ref{fig: model}, where the transmitter is powered by an exogenous stochastic energy arrival process $E_t$ stored in a battery of size $B_{max}$. The energy of the transmitted symbol at each channel use is limited by the available energy in the battery, which unlike in traditionally powered communication systems is a random quantity due to the randomness in the energy harvesting process $E_t$ and moreover depends on the energy expenditure in the previous channel uses. Understanding the capacity of such newly emerging communication systems and the optimal principles to design and operate them has received significant attention over the recent years \cite{YangUlukus, TutuncuogluYener, OzelUlukus_BInfty,OzelTutuncuoglu_JSAC,MahdaviYates}.

In the limiting case $B_{max} = \infty$, the capacity of the energy harvesting communication system in Fig.~\ref{fig: model} has been characterized by Ozel and Ulukus in \cite{OzelUlukus_BInfty}. Their result shows that in this asymptotic case the capacity of the energy-harvesting system is equal to the capacity of a classical AWGN channel with an average power constraint equal to the average energy harvesting rate  $\E[E_t]$. Perhaps even more importantly, the result of \cite{OzelUlukus_BInfty} offers a number of important engineering insights for the large battery limit: first it shows that, via simple modifications,  the standard communication and coding techniques developed for the classical AWGN channel can be used to achieve the capacity of an energy-harvesting AWGN channel; second, it shows that, in this asymptotic case, the only relevant property of the energy-harvesting process in determining capacity is the average energy harvesting rate. Two energy harvesting mechanisms are equivalent  as long as  they have the same average energy arrival rate.

\begin{figure}[t!]
\centering
\scalebox{0.5}
{\includegraphics{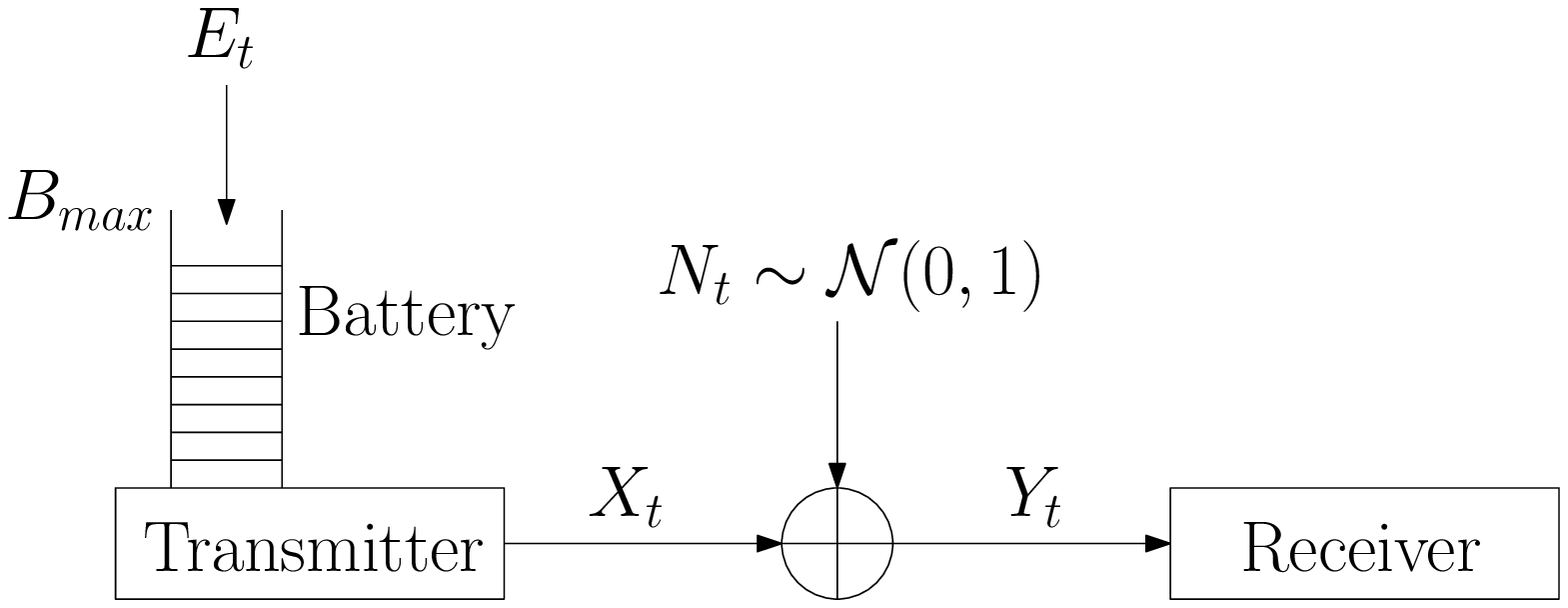}}
\caption{System model for an energy-harvesting AWGN communication system.}
\label{fig: model}
\end{figure}

Can we obtain analogous engineering insights for the more realistic case of finite battery? For example, how does the capacity of the energy-harvesting AWGN channel in Fig.~\ref{fig: model} depend on  major system parameters such as $B_{max}$ and $E_t$? Are there different operating regimes where this dependence is qualitatively different? Given a communication system powered with a certain energy harvesting mechanism $E_t$, how can we optimally choose the battery size $B_{max}$? What are the properties of the energy harvesting process $E_t$ that are critical to capacity in the finite battery regime? Consequently, what are more desirable and less desirable energy-harvesting profiles? These are foremost engineering questions, the answers of which can guide the design of optimal communication architectures for such systems.

Despite significant recent effort \cite{MaoHassibi,DeekshithSarma,TutuncuogluOzel_ISIT} to characterize the capacity of the energy-harvesting channel in Fig.~\ref{fig: model} in the finite battery case, we currently lack an understanding of the above questions. For example, \cite{MaoHassibi} provides a formulation of the capacity in terms of the Verdu-Han general framework and based on a conjecture on the properties of the optimal energy management strategy derives a lower bound to the capacity which is numerically computable for a given setup. However, it is difficult to obtain the above high-level insights from the numerical evaluations. Indeed, even in the case of zero battery, $B_{max}=0$, where \cite{OzelUlukus_Bzero} provides an exact single letter characterization of the capacity  as an optimization problem over the so called Shannon-strategies, the resultant optimization is difficult to solve and requires numerical evaluations, therefore providing limited high-level insights.

\subsection{Overview of Our Results}

In this paper, we take an alternative approach. Instead of seeking the exact capacity, we seek to provide a simple approximation to the capacity (with bounded guarantee on the approximation gap) which can provide insights on the above engineering questions. As a starting point, we concentrate on an i.i.d. Bernoulli energy arrival process, i.e. $E_t=E$ with probability $p$ and zero otherwise, and is independent across different channel uses. We show that in this case the capacity of the energy harvesting AWGN channel in  Fig.~\ref{fig: model} is approximately given by
\begin{equation}\label{eq:approxcap}
C\approx\left\{ \begin{array}{ll}
\frac{1}{2}\log(1+pB_{max}) & \mbox{when $B_{max} \leq E$} \\
\frac{1}{2}\log(1+pE) & \mbox{when $B_{max} > E$}.
\end{array}  \right.
\end{equation}   
The approximation gap is bounded by $2.58$ bits  for all values of the system parameters $p, E$, and $B_{max}$. See Fig.~\ref{fig: CapBound}.

The  capacity approximation in \eqref{eq:approxcap} provides couple of important insights: First, it identifies the dependence of the capacity to major system parameters. There are two regimes where this dependence is qualitatively different: in the large battery regime ($B_{max}>E$), the capacity is mainly determined by the average energy arrival rate and is (almost) independent of the exact value of the battery size; on the other hand, in the small battery regime ($B_{max}<E$), the capacity depends critically on the battery size and is increasing logarithmically with increasing $B_{max}$. The formula also suggests that choosing $B_{max}\approx E$ is sufficient to extract most of the capacity of the system (achieved at $B_{max}=\infty$).  
One can also observe that while in the large battery regime the only property of the energy harvesting profile that impacts  capacity is the average rate, in the small battery regime energy profiles that are less peaky over time lead to larger capacity. See Figure~\ref{fig: CapBound_2} which compares two energy profiles with the same average rate.

The main ingredient of the above result is a near optimal online strategy we develop for 
energy/power control over the Bernoulli energy harvesting channel. The problem of optimal energy/power allocation over an energy harvesting channel has been extensively studied in the literature. While the offline version is well-understood \cite{YangUlukus,TutuncuogluYener,OzelTutuncuoglu_JSAC}, the online version of the problem is known to be difficult \cite{OzelTutuncuoglu_JSAC,KhuzaniSaffar,MaoKoksalShroff,SrivastavaKoksal} and several works suggest simple online power control policies with and without performance guarantees. For example, \cite{SrivastavaKoksal} shows that a simple strategy that allocates constant energy, equal to the average energy arrival rate,  to each channel use as long the battery is not empty becomes asymptotically optimal as $B_{max}\to\infty$. However as we show in Section~\ref{sec: simulations}, this strategy can be arbitrarily away from optimality for finite values of $B_{max}$. In contrast, we show analytically that the online energy allocation policy we propose remains within $0.973$ bits of the optimal value for all values of the system parameters. This result can be of interest in its own right.

In the final section of the paper, we extend our approximation results to general i.i.d. energy arrival processes where each $E_t$ is i.i.d. according to some arbitrary distribution. We show that a simple modification allows to apply the near optimal energy allocation policy and the information-theoretic coding strategy we develop in the earlier sections for the Bernoulli process to general i.i.d. energy arrival processes. We show that for many distributions of $E_t$, this yields an approximate characterization of the capacity within a constant gap, though with a larger constant. However, we also illustrate that one can engineer specific distributions for which our approach would fail to provide a constant gap approximation. 

The constant gap approach we propose in this paper is most useful to understand the capacity of energy harvesting communication channels operating at moderate to high SNRs. This, for example, can be the operating regime of a base station in a rural area powered by renewable energy sources (ex. wind or solar). However, the insights obtained from such an analysis can be applicable even at low SNR. Note that the $2.58$ bits/s/Hz  is an analytical upper bound on the worst case gap over all SNR regimes and the actual gap between the rate achieved by the
strategies we propose and the true capacity of the system can be much smaller especially at low SNR. More generally, we believe the approximation philosophy we propose
in this paper will be useful in providing a basis for comparing the performances of different strategies,
developing insights into the capacity of the system and giving a sense of the remaining gap in
characterizing the problem, when obtaining an exact \emph{insightful} capacity formula proves to be difficult.
\begin{figure}[t!]
\centering
\scalebox{0.7}
{\includegraphics{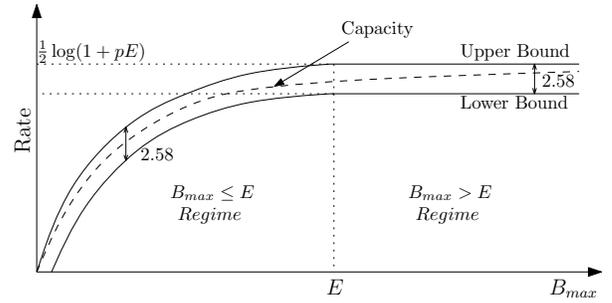}}
\caption{Illustration of the main result: constant gap upper and lower bounds on the capacity of the system. The plot shows how the capacity and its upper and lower bounds vary with $B_{max}$ for a fixed pair of $(p, E)$.}
\label{fig: CapBound}
\end{figure}

\begin{figure}[t!]
\centering
\scalebox{0.7}
{\includegraphics{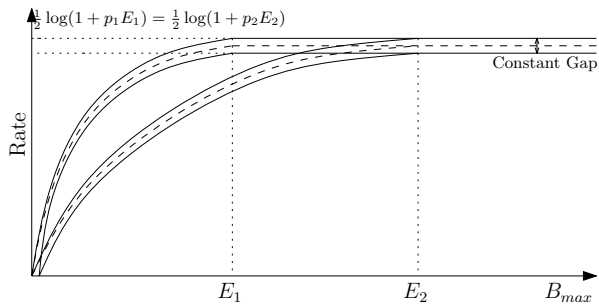}}
\caption{Comparison of the capacity for two energy harvesting profiles $(p_1, E_1)$ and $(p_2, E_2)$ with the same average rate $p_1E_1=p_2E_2$ ($p_1>p_2, E_1<E_2$).}
\label{fig: CapBound_2}
\end{figure}

\section{System Model}
\label{sec: SysModel}

We consider an AWGN communication system powered by an energy harvesting mechanism with limited battery (See Fig.~\ref{fig: model}). At each time step $t$, $E_t$ amount of energy is harvested from the exogenous energy source and is stored in the battery of size $B_{max}$. In the case the harvested energy $E_t$ exceeds the available space in the battery at time $t$, the battery is charged to maximum capacity and the remaining energy is discarded. Let $X_t$ denote the scalar real input to the channel at time $t$, $Y_t$ the  output of the channel and $N_t$ the additive noise with unit normal distribution $\mathcal{N}(0,1)$. We have
\begin{equation}
Y_t=X_t +N_t.
\end{equation}
Let $B_t$ represent the amount of energy available in the battery for transmission at time step $t$. Then the  transmitted signal $X_t$ is amplitude constrained by the available energy in the battery $B_t$. The system energy constraints can be summarized as follows:
\begin{align}
|X_t|^2 & \le B_t \label{eq: ampConst}\\
B_{t+1} & = \min(B_t + E_{t+1} - |X_t|^2, B_{max})  \label{eq: batteryUpdate}
\end{align}
Equation \eqref{eq: ampConst} represents the amplitude constraint on the input, and \eqref{eq: batteryUpdate} represents the update rule for the energy available in the battery. \footnote{Here, we require the harvested energy be stored in the battery before it can be used for transmission. An alternative energy storage model, considered in \cite{OzelUlukus_Bzero,MaoHassibi}, is to allow the harvested energy $E_t$ to be used instantaneously at time $t$ and the remaining amount be stored in the battery. 
The capacities under the two energy storage models are different in general.}

The harvested energy at each time step $E_t$ is a discrete-time stochastic process dictated by the energy harvesting mechanism. We assume that the energy arrival process  is causally known at the transmitter (i.e. at time $t$ the transmitter knows $\{E_t, E_{t-1}, \dots\}$), but not at the receiver. In the sequel, we first focus on a simple Bernoulli process (Sections~\ref{sec: Online_Policy} and \ref{sec: Info_Cap}), where $E_t$'s are i.i.d. Bernoulli random variables:
\begin{align}
\label{eq: Et}
E_t = \left\{ \begin{array}{rl}
E & \mbox{w.p. $p$} \\
0 & \mbox{w.p. $1-p$},
\end{array}  \right.
\end{align}
so at each time step, either an energy packet of size $E$ is harvested, or no energy is harvested at all, independent of the other time steps. We then extend our results to more general i.i.d. energy harvesting processes in Section~\ref{sec: General_EP}, where we assume that $E_t$ are i.i.d. with common cdf $F(x)$. 

The information-theoretic capacity of the channel is defined in the usual way as the largest rate at which the transmitter can reliably communicate to the receiver under the  system energy constraints in \eqref{eq: ampConst} and \eqref{eq: batteryUpdate} and the assumption that the realizations of $E_t$ are known only to the transmitter in a causal fashion. While the main focus of our paper is the information-theoretic capacity of this channel, in Section~\ref{sec: Online_Policy} we also study a related problem, optimal (online) energy allocation for maximizing the average rate (or utility) of an energy-harvesting system. This problem is defined more precisely in the corresponding section and forms the basis for the information-theoretic results we prove in Section~\ref{sec: Info_Cap}.

\section{Main Result}
\label{sec: main}

The following theorem is the main result of the paper. 
\begin{thm}[Main Result]
\label{thm: MainResult}
The capacity $C$ of the channel described in Section \ref{sec: SysModel} with Bernoulli energy arrival process $E_t$ satisfies

\begin{align*}
\frac{1}{2} \log  \left(1 + p \cdot \min\{  B_{max}, E\} \right)& - 2.58 \le  \; C \\
 \le \frac{1}{2} \log & \left( 1 + p \cdot \min\{B_{max}, E \} \right).
\end{align*}
\end{thm}

The bound in Theorem \ref{thm: MainResult} is illustrated in Figure \ref{fig: CapBound}. The result shows that the capacity of a finite battery energy harvesting system is within 2.58 bits of $\frac{1}{2} \log \left( 1 + p \cdot \min\{B_{max}, E \} \right)$ for any choice of the system parameters $p, E$ and $B_{max}$. We therefore refer to $\frac{1}{2} \log \left( 1 + p \cdot \min\{B_{max}, E \} \right)$ as the approximate capacity of the channel. Figure \ref{fig: CapBound_2} compares the approximate capacity under two different Bernoulli energy harvesting profiles. 

The main ingredient of the above theorem is a near optimal online energy allocation policy we develop for the Bernoulli energy harvesting channel in Section~\ref{sec: Online_Policy}. We show that this strategy is within $0.973$ bits of optimality for all values of the system parameters. We state the corresponding results in Theorems~\ref{thm: const_gap_rate} and \ref{thm: const_gap_rate2}, after we precisely define the online energy allocation problem in Section~\ref{sec: Online_Policy}. Finally, we discuss how our results from Sections \ref{sec: Online_Policy} and \ref{sec: Info_Cap} for the Bernoulli process can be extended to more general i.i.d. energy harvesting profiles in Section \ref{sec: General_EP}.

\section{A Near Optimal Online Policy for Energy Control}
\label{sec: Online_Policy}

In this section, we study an online energy allocation problem for the energy harvesting communication system with finite battery described in Section \ref{sec: SysModel}. The near optimal energy control policy we develop in this section turns out to be a critical ingredient of the approximate capacity characterization we develop in the next section.

We consider the energy harvesting transmitter described in Section \ref{sec: SysModel} and assume that we are specified an energy-rate (or energy-utility) function $r(g)$ which specifies the rate or utility $r(g(t))$ obtained at a given channel use $t$ as a function of the energy $g(t)$ allocated for transmission at time $t$.\footnote{The continuous-time version of this problem has been considered in many works in the literature \cite{YangUlukus,TutuncuogluYener,OzelTutuncuoglu_JSAC,KhuzaniSaffar,MaoKoksalShroff,SrivastavaKoksal} where it is more commonly referred to as the power control problem and the function $r(g)$ is called the power-rate or the power-utility function in this case. Note that in our discrete-time setup power corresponds to energy per channel use, and therefore the terms energy (per channel use) and power can be used interchangeably. We prefer to refer to the problem as the energy allocation problem.} We assume the energy arrival process is a Bernoulli  process described by \eqref{eq: Et}, and again only online, i.e. causal, information of the energy packet arrivals is available at the transmitter.  An online policy $g(t)$ denotes the amount of energy transmitter decides to allocate for transmission at time $t$. We call $g(t)$ a feasible online policy, if $g(t)$ satisfies:
\begin{align}
& 0 \le g(t)  \le B_t \label{eq: allocConst}\\
& B_{t+1}  =  \min(B_t + E_{t+1} -g( t) , B_{max})  \label{eq: batteryUpdate2}
\end{align}
Notice that constraints \eqref{eq: allocConst} and \eqref{eq: batteryUpdate2} are analogous to the system energy constraints \eqref{eq: ampConst} and \eqref{eq: batteryUpdate} from Section \ref{sec: SysModel}, except $|X_t|^2$, the amount of energy used for transmission at time step $t$, from Section \ref{sec: SysModel} is replaced with $g(t)$ here. Moreover, the energy allocated at time $t$ can only depend on the past realizations of the energy harvesting process, i.e. we have the causality constraint
\begin{equation}
\label{eq:causality}
g(t)=f(t,\{E_{i}\}_{i=0}^t).
\end{equation}
Since the energy arrival process $E_t$ is a stochastic process, the quantity $g(t)$ is random. Let $\mathcal{G}$ denote the class of online policies satisfying constraints \eqref{eq: allocConst}, \eqref{eq: batteryUpdate2} and \eqref{eq:causality}. Then, our goal is to maximize the long term average rate over the class of feasible online policies:

\begin{align}
 \label{eq: obj}
\max_{g \in \mathcal{G}} \liminf_{N \to \infty}  \mathbb{E}\left[ \frac{1}{N} \displaystyle\sum\limits_{t = 1}^N r(g( t)) \right],
\end{align}
where the expectation here is over the ensembles of $\{E_t\}_{t=0}^{N}$ and $r(\cdot)$ is the energy-rate function of interest. Since our main motivation for studying this problem is to use the solution to approximate the capacity of the AWGN  channel in the next section, we restrict our attention to the following classical AWGN rate function:
\begin{align}
\label{eq: AWGN_rate}
r(x) = \frac{1}{2} \log\left(1 + x \right)
\end{align}
in bits per channel use. Therefore, the problem of interest here is:

\begin{align}
 \label{eq: AWGN_cap_obj}
\max_{g \in \mathcal{G}} & \liminf_{N \to \infty}  \mathbb{E} \left[ \frac{1}{N} \displaystyle\sum\limits_{t = 1}^N \frac{1}{2} \log\left( 1 + g( t)\right) \right]. 
\end{align}
\subsection{The case $B_{max} \le E$} 
\label{subsec: rate_B_less_E}

We first analyze the case where $B_{max} \le E$. In this case, according to our system definition, every time a non-zero energy packet arrives, the battery is charged to full, and the left over energy is wasted. Since, the system is reset to the initial state of full battery each time a non-zero energy packet arrives, each epoch (the period of time instances between two adjacent non-zero packet arrivals) is independent and statistically indistinguishable from every other epoch. Motivated by this observation, we propose to use a strategy where the energy $g(t)$ allocated to transmission at time $t$ depends only on the number of channel uses since the last time the battery was recharged, i.e., $g(t) = \tilde{g}(j)$ where $j = t - \max\{t'\le t: E_{t'} = E \}$, for a function 
$\tilde{g}(j)$ that satisfies
\begin{align}
\label{eq: constraint}
\displaystyle\sum\limits_{j=0}^\infty \tilde{g}(j) \le B_{max} \;\;\;\;  \text{and} \;\;\;\ \tilde{g}(j) \ge 0 \;\;\; \forall j. 
\end{align}
Note that an energy allocation policy that satisfies the above properties clearly satisfies the feasibility constraints \eqref{eq: allocConst} and \eqref{eq: batteryUpdate2}. Moreover, it  uses only information about the past realizations of the process. We choose $\tilde{g}(j) = p(1-p)^j B_{max}$, which clearly satisfies \eqref{eq: constraint}, and show in the following theorem that it achieves an objective value that's no more than $0.973$ bits away from the optimum  value of the optimization problem given in \eqref{eq: AWGN_cap_obj} for any value of $p$ and  $B_{max}$. 

\begin{thm}
\label{thm: const_gap_rate}
Let $g^\prime(t) = \tilde{g}(j)$ where $j = t - \max\{t'\le t: E_{t'} = E \}$ and $\tilde{g}(j) =  p(1-p)^j B_{max}$ for $j=0,1,2,...$. When $B_{max}\leq E$,  we have the following guarantee:
\begin{align}
\liminf_{N \to \infty}&\; \mathbb{E} \left[ \frac{1}{N} \displaystyle\sum\limits_{t = 1}^N \frac{1}{2} \log\left( 1 + g^\prime(t)\right) \right] \notag \\
 & \ge  \max_{g \in \mathcal{G}}  \liminf_{N \to \infty} \mathbb{E} \left[  \frac{1}{N} \displaystyle\sum\limits_{t = 1}^N \frac{1}{2} \log\left( 1 + g( t)\right) \right] - 0.973. \label{eq: const_gap_rate}
\end{align}
\end{thm}

\begin{proof}
The detailed proof of the theorem is deferred to the Appendix. The proof follows roughly three major steps. 
\begin{itemize}
\item{Using Jensen's inequality, we first show that 
\begin{align}
\max_{g \in \mathcal{G}} \liminf_{N \to \infty} \mathbb{E} & \left[ \frac{1}{N}  \displaystyle\sum\limits_{t = 1}^N \frac{1}{2} \log\left( 1 + g( t) \right) \right]  \notag\\
& \leq \frac{1}{2}\log(1+pB_{max})\label{eq: policy_upper}
\end{align}
}
\item{Using the fact that $g^\prime(t)$ is same across different epochs, we turn our achievable rate (left-hand-side of \eqref{eq: const_gap_rate}) into the following expression as a function of $\tilde{g}(j)$:
\begin{align}
\label{eq: ach}
\displaystyle\sum\limits_{j = 0}^{\infty}p &(1-p)^j  \frac{1}{2} \log(1 + \tilde{g}(j)) \notag \\
& = \displaystyle\sum\limits_{j = 0}^{\infty} p(1-p)^j \frac{1}{2} \log(1 + p(1-p)^j B_{max}) 
\end{align}
}
\item{Finally, we upper bound the gap between the objective value achieved with $g'(t)$, i.e. Equation \eqref{eq: ach}, and $\frac{1}{2}\log(1+pB_{max})$ by a constant.}
\end{itemize}
\end{proof}

\begin{figure}[t]
\centering
\scalebox{0.53}
{\includegraphics{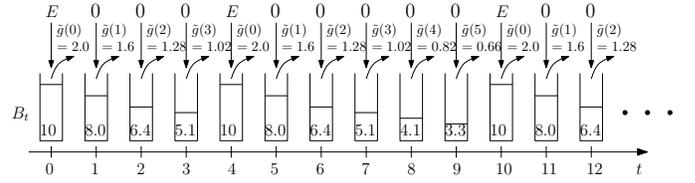}}
\caption{Illustration of the energy allocation policy in Theorem \ref{thm: const_gap_rate}. The parameters in this example are $p = 0.2$, $B_{max} = 10$, $E \ge B_{max}$, so our policy allocates $\tilde{g}(j) = 2\cdot 0.8^j$ for $j = 0,1,2, ...$ Note that the energy control policy is reset to $j=0$ each time a packet arrives.}
\label{fig: EnergyAlloc}
\end{figure}

Fig.~\ref{fig: EnergyAlloc} illustrates the energy allocation policy in Theorem ~\ref{thm: const_gap_rate}. In this allocation policy we use $p$ fraction of the remaining energy in the battery at each time (so the energy in the battery decays like $B_{t} = (1-p)^j B_{max}$). The motivation for this energy allocation policy is the following: for the Bernoulli arrival process $E_t$, the inter-arrival time is a Geometric random variable with parameter $p$. We know that the Geometric random variable is memoryless and has mean $1/p$. Therefore, at each time step, the expected number of time steps to the next energy arrival is $1/p$. Furthermore, since $\log(\cdot)$ is a concave function, results from \cite{YangUlukus, TutuncuogluYener} tell us that in order to achieve higher rate, we want to allocate the energy as uniform as possible between energy arrivals, i.e. if the current energy level in the battery is $B_t$ and we knew that the next recharge of the battery would be in exactly $m$ channel uses, we would allocate $B_t/m$ energy to each of the next $m$ channel uses. For the online case of interest here, we do not know when the next energy arrival would be. Instead, we use  the expected time to the next energy arrival as a basis: since at each time step, the expected time to the next energy arrival is $1/p$, we use a fraction $p$ of the currently available energy. 

Some simple online policies with and without performance guarantees have been earlier proposed in \cite{OzelTutuncuoglu_JSAC,MaoKoksalShroff,SrivastavaKoksal}. None of these strategies utilize the idea of exponential energy usage we propose here and can achieve the optimal rate within a constant gap uniformly over all parameter ranges. In Section~\ref{sec: simulations} below, we provide simulations which illustrate that these strategies can be arbitrarily away from optimality. However, before that we first address the remaining case of $B_{max}> E$.  

\subsection{The case $B_{max} > E$} 
\label{subsec: rate_B_greater_E}

Note that when $B_{max} \leq E$, because the energy must be stored into the battery before it can be used, the extra energy is wasted, and the system $B_{max} \leq E$ is equivalent to a system where the energy packet size is exactly equal to $B_{max}$. Therefore, the average optimal rate we can achieve with online energy management strategies is independent of $E$, and  Theorem~\ref{thm: const_gap_rate} characterizes this rate as approximately $\frac{1}{2}\log(1+pB_{max})$. In the case $B_{max} > E$, we show that the average optimal rate is given by $\frac{1}{2}\log(1+pE)$ which can be achieved by a simple modification of the energy control policy proposed in the earlier section. We have the following theorem.

\begin{thm}
\label{thm: const_gap_rate2}
Let $g^\prime(t) = \tilde{g}(j)$ where $j = t - \max\{t'\le t: E_{t'} = E \}$ and $\tilde{g}(j) =  p(1-p)^j E$ for $j=0,1,2,...$. When $B_{max}> E$, we have the following guarantee:
\begin{align}
\liminf_{N \to \infty}&\; \mathbb{E} \left[ \frac{1}{N} \displaystyle\sum\limits_{t = 1}^N \frac{1}{2} \log\left( 1 + g^\prime(t)\right) \right] \notag \\
 & \ge  \max_{g \in \mathcal{G}}  \liminf_{N \to \infty} \mathbb{E} \left[ \frac{1}{N} \displaystyle\sum\limits_{t = 1}^N \frac{1}{2} \log\left( 1 + g( t)\right) \right] - 0.973. \label{eq: const_gap_rate_2}
\end{align}
\end{thm}

Note that in the case $B_{max}> E$, each time an energy packet arrives $B_t$, the available energy in the battery, becomes at least as large as $E$. Based on this observation, our strategy utilizes a fraction $p$ of $E$ when $j=0$, i.e. if the energy packet has just arrived in the current channel use; a fraction $p$ of the remaining $(1-p)E$ in the next channel use $j=1$, i.e. if the energy has arrived in the previous channel use and not the current one, etc. It is easy to verify that since 
\begin{align}
\label{eq:needeqnumber}
\displaystyle\sum\limits_{j=0}^\infty \tilde{g}(j) \le E \;\;\;\;  \text{and} \;\;\;\ \tilde{g}(j) \ge 0 \;\;\; \forall j,
\end{align} 
this strategy satisfies the energy  feasibility constraints in \eqref{eq: allocConst} and \eqref{eq: batteryUpdate2} when $B_{max}> E$. Indeed, this energy management policy is quite conservative and can be clearly wasteful of resources. Consider the first epoch which starts with the arrival of the first energy packet: the remaining energy in the battery $j$-channel uses after the first energy packet arrives is given by $(1-p)^j E$, so at the time the second  packet arrives, there will be some residual energy left in the battery, at least part of which will add up to the arriving energy packet $E$ since $B_{max}>E$. The strategy we propose ignores this residual energy. An equivalent way of thinking about our strategy is that it operates as if $B_{max}=E$ even though $B_{max}>E$. However this strategy still turns out to be within constant number of bits of the optimal value. One immediate way to improve this strategy would be to use, at each time step, a fraction $p$ of the currently available energy in the battery $B_t$. However, this improved strategy turns out to be difficult to analyze analytically. In the next section, we present simulation results which demonstrate the improvement due to this modification.

Theorem~\ref{thm: const_gap_rate2} can be proved by using similar lines to the proof of Theorem~\ref{thm: const_gap_rate}. In the appendix, again based on Jensen's inequality, we show that 
\begin{align*}
\max_{g \in \mathcal{G}}  \liminf_{N \to \infty} \mathbb{E} \left[ \frac{1}{N} \displaystyle\sum\limits_{t = 1}^N \frac{1}{2}  \log\left( 1 + g( t)\right) \right]
 \leq \frac{1}{2}\log(1+pE)
\end{align*}
in this case. The remaining step is to show is that the strategy we propose in Theorem~\ref{thm: const_gap_rate2} is within $0.973$ bits of this upper bound. While this can be shown from first principles by following the steps of Theorem~\ref{thm: const_gap_rate}, it can also be directly observed from the proof of  Theorem~\ref{thm: const_gap_rate}: Fix the energy packet size $E$. When the battery size is as large as $E$, i.e. $B_{max}=E$, the strategy in Theorem~\ref{thm: const_gap_rate} reduces to the strategy in Theorem~\ref{thm: const_gap_rate2}. Moreover, the proof of  Theorem~\ref{thm: const_gap_rate} establishes that this strategy achieves $ \frac{1}{2}\log(1+pE)$ within $0.973$ bits. Now in a system with same $E$ but larger $B_{max}$, so that $B_{max}>E$, the same strategy is still feasible and will clearly achieve the same rate. This argument shows that when $B_{max}> E$, the strategy proposed in Theorem~\ref{thm: const_gap_rate2} achieves an average rate at least as large as $ \frac{1}{2}\log(1+pE)-0.973$ bits/channel use.  

\subsection{Numerical Evaluations} 
\label{sec: simulations}

\begin{figure*}[ht]
 \centering
 \subfigure[$B_{max} \le E$]{
  \includegraphics[scale=0.40]{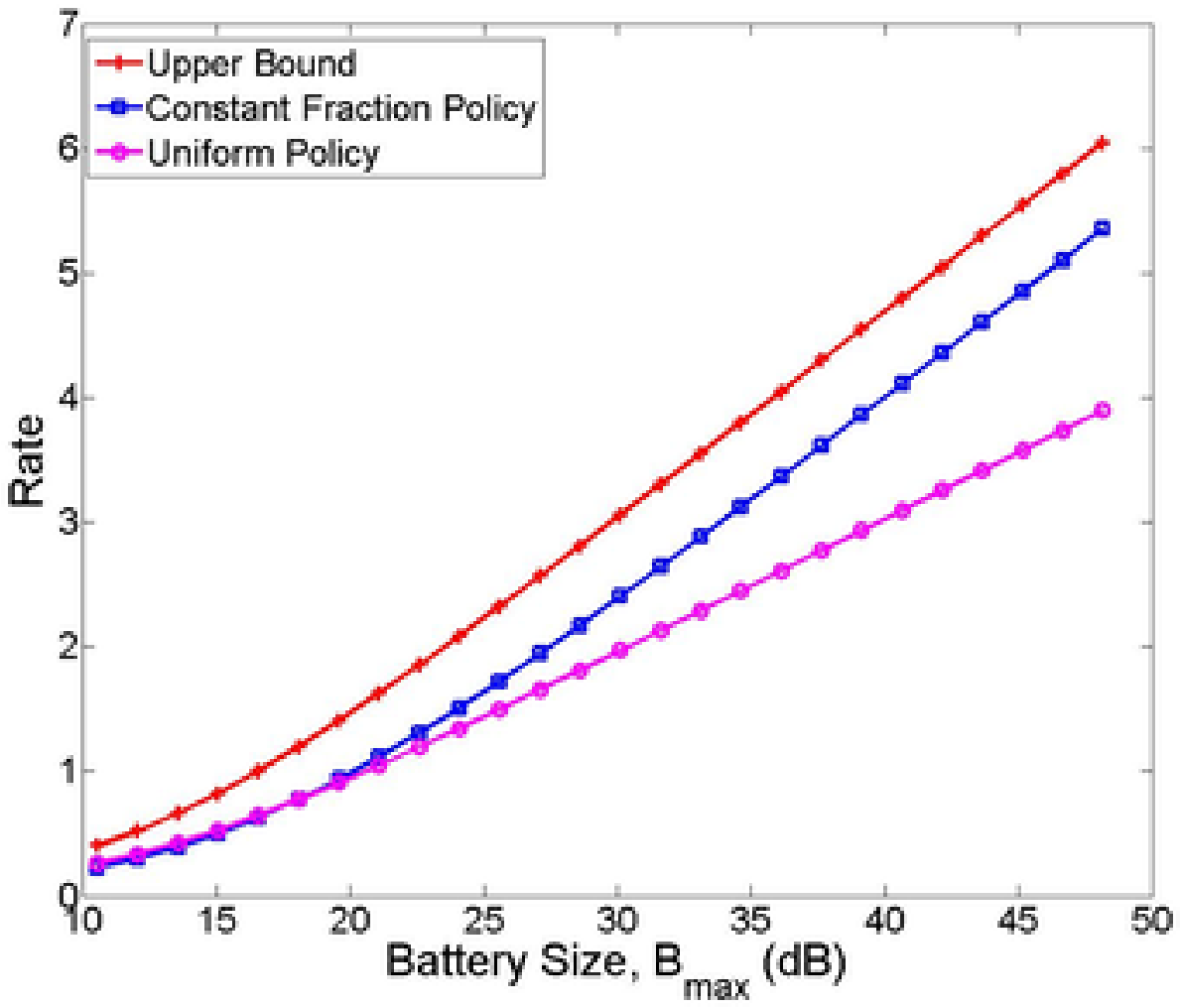}
   \label{fig: B_le_E_15}
   }
 \subfigure[$B_{max} = 2E$]{
  \includegraphics[scale=0.40]{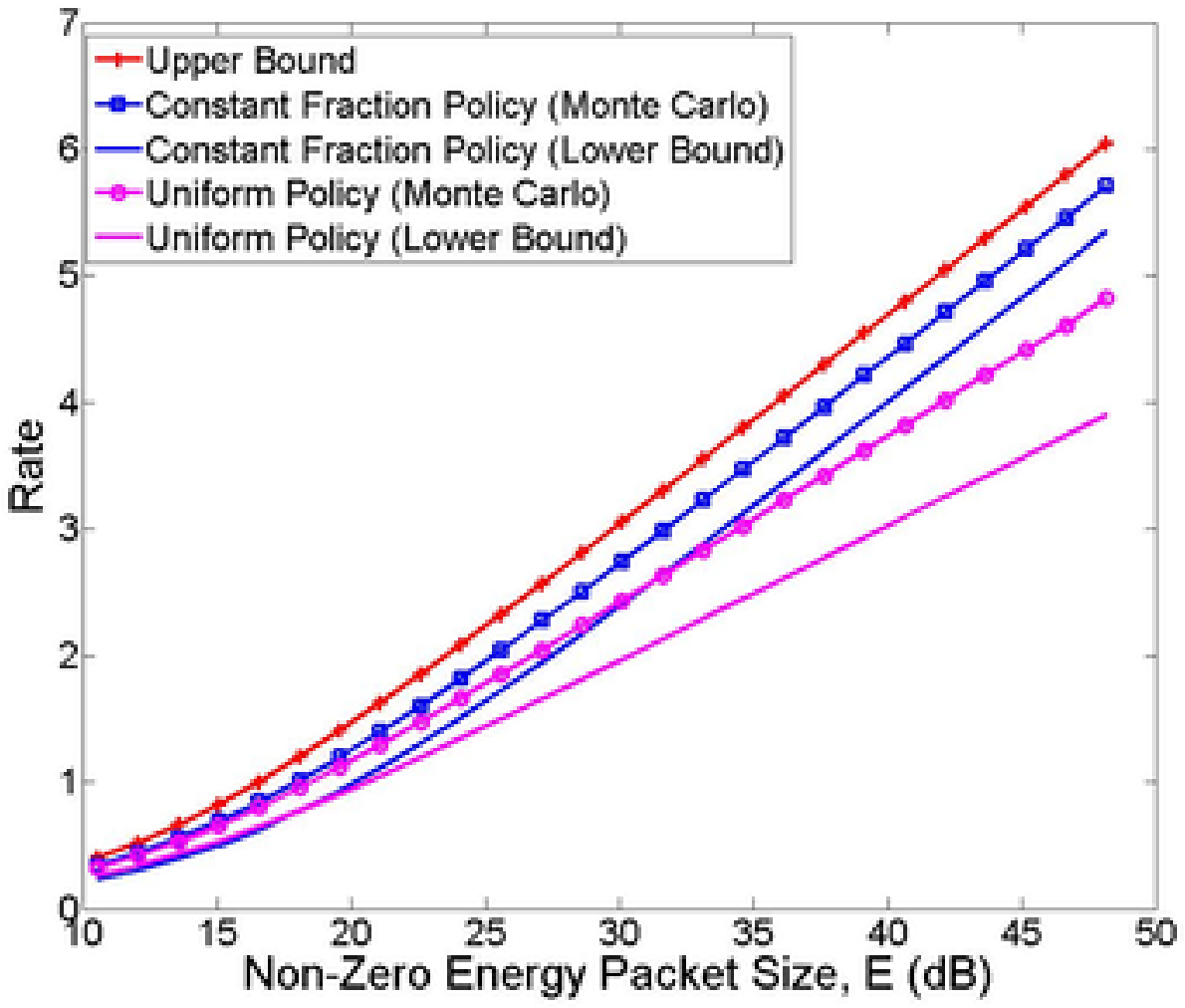}
   \label{fig: B_2E_15}
   }
 \subfigure[$B_{max} = 8E$]{
  \includegraphics[scale=0.40]{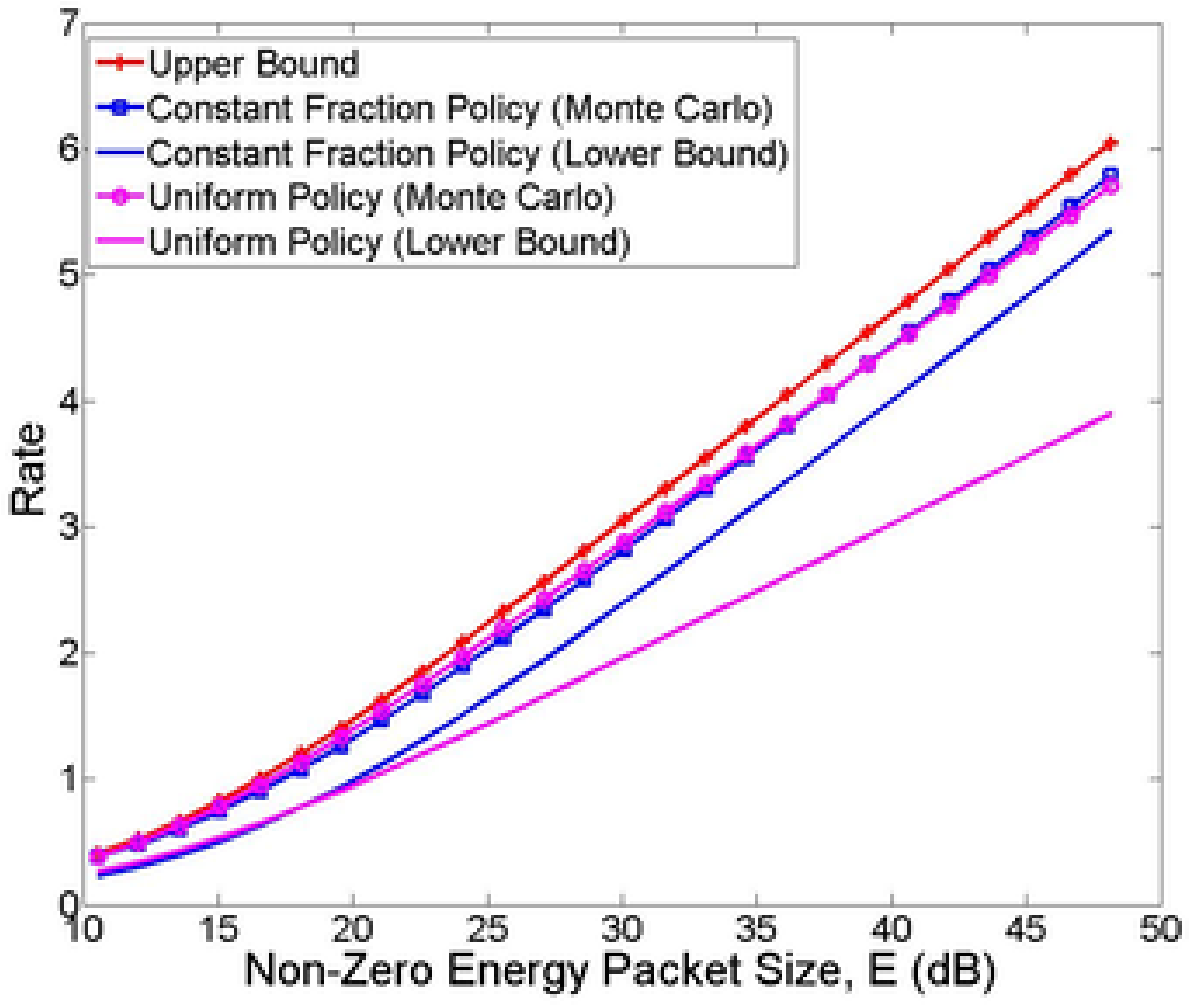}
     \label{fig: B_8E_15}
   }
 \caption{Numerical evaluations for a system described in Section \ref{sec: SysModel} with energy arrival probability $p = 1/15$.}
  \label{fig: p15}
\end{figure*}

\begin{figure*}[ht]
 \centering
 \subfigure[$B_{max} \le E$]{
  \includegraphics[scale=0.40]{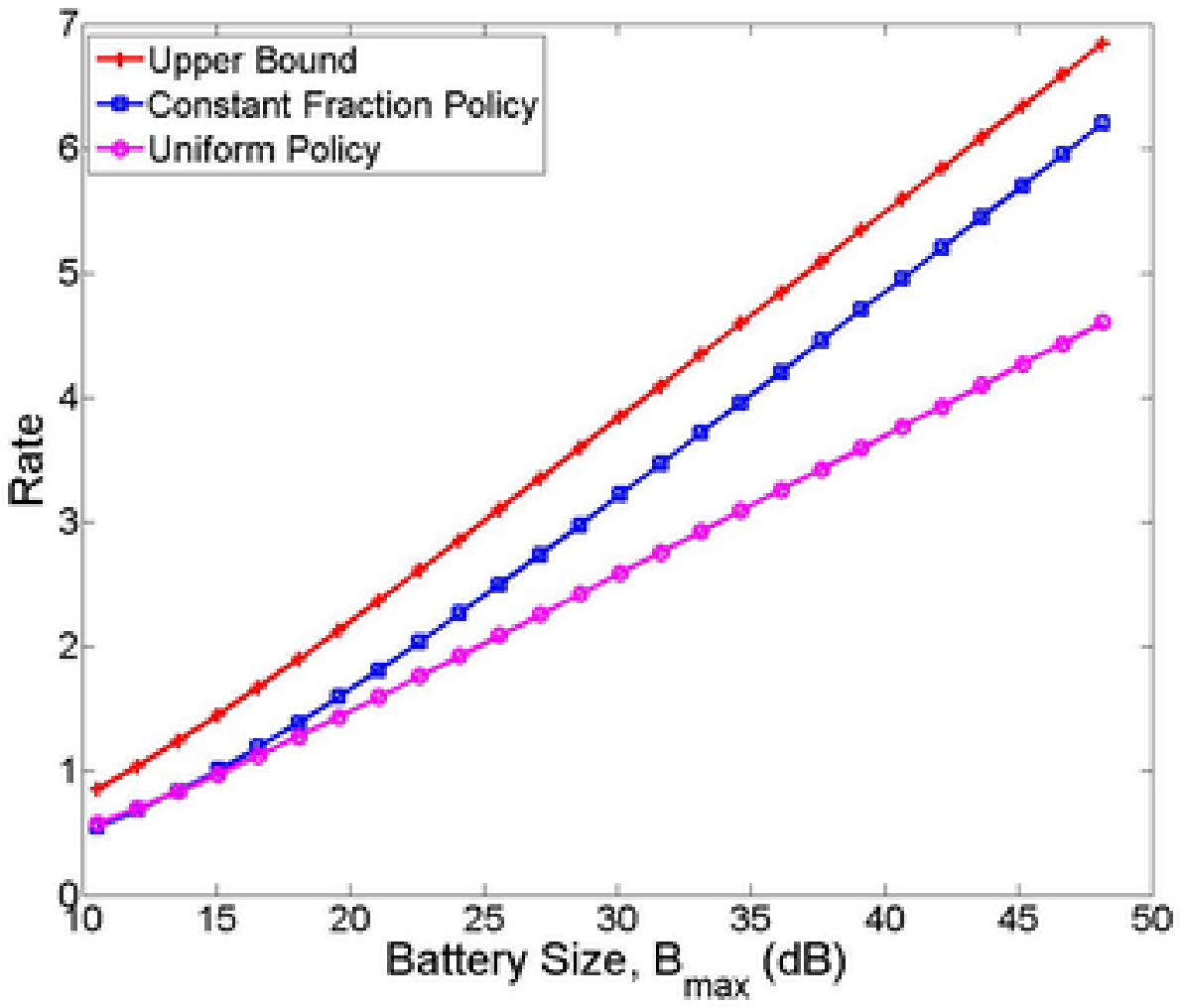}
   \label{fig: B_le_E_5}
   }
 \subfigure[$B_{max} = 2E$]{
  \includegraphics[scale=0.40]{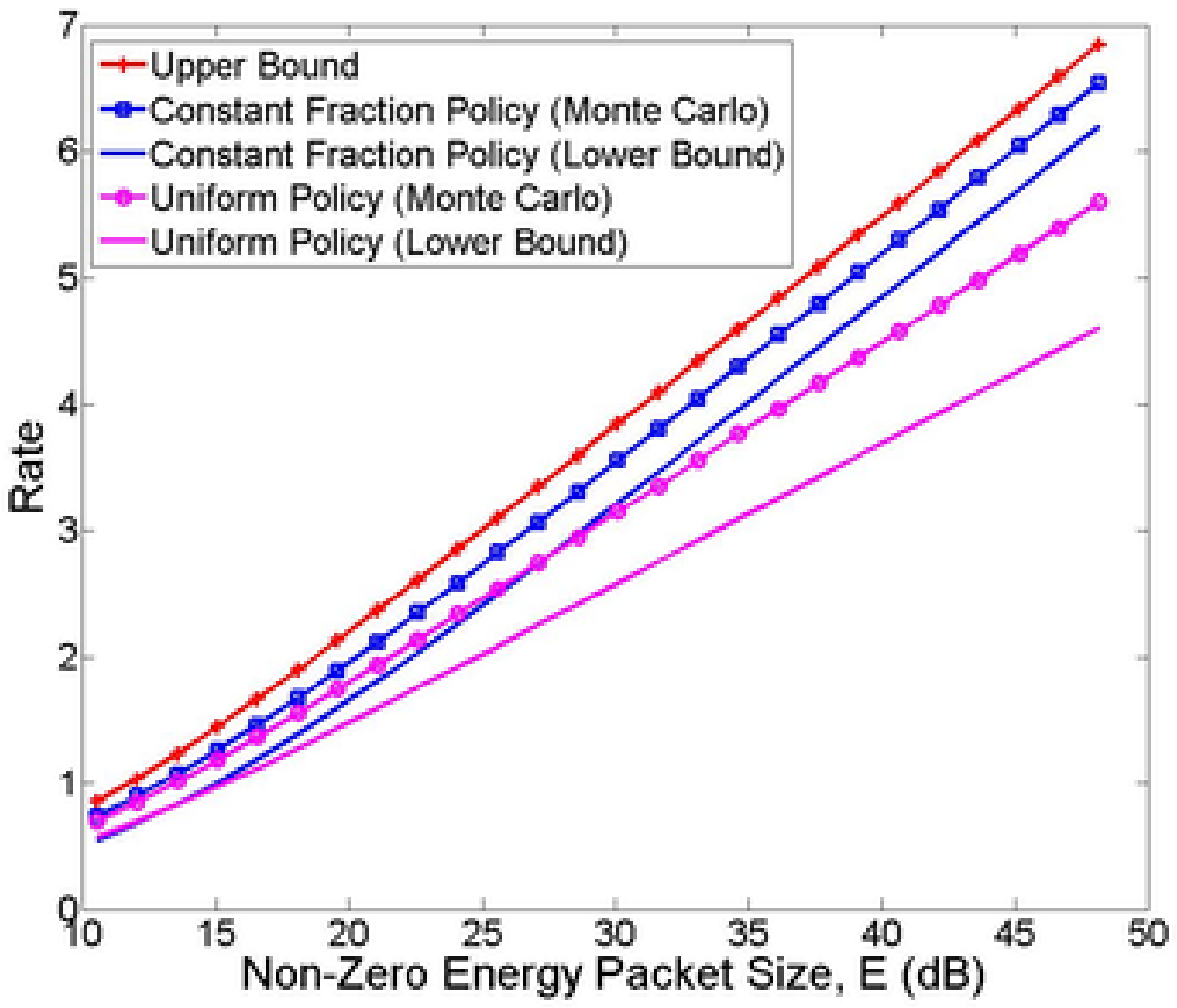}
   \label{fig: B_2E_5}
   }
 \subfigure[$B_{max} = 8E$]{
  \includegraphics[scale=0.40]{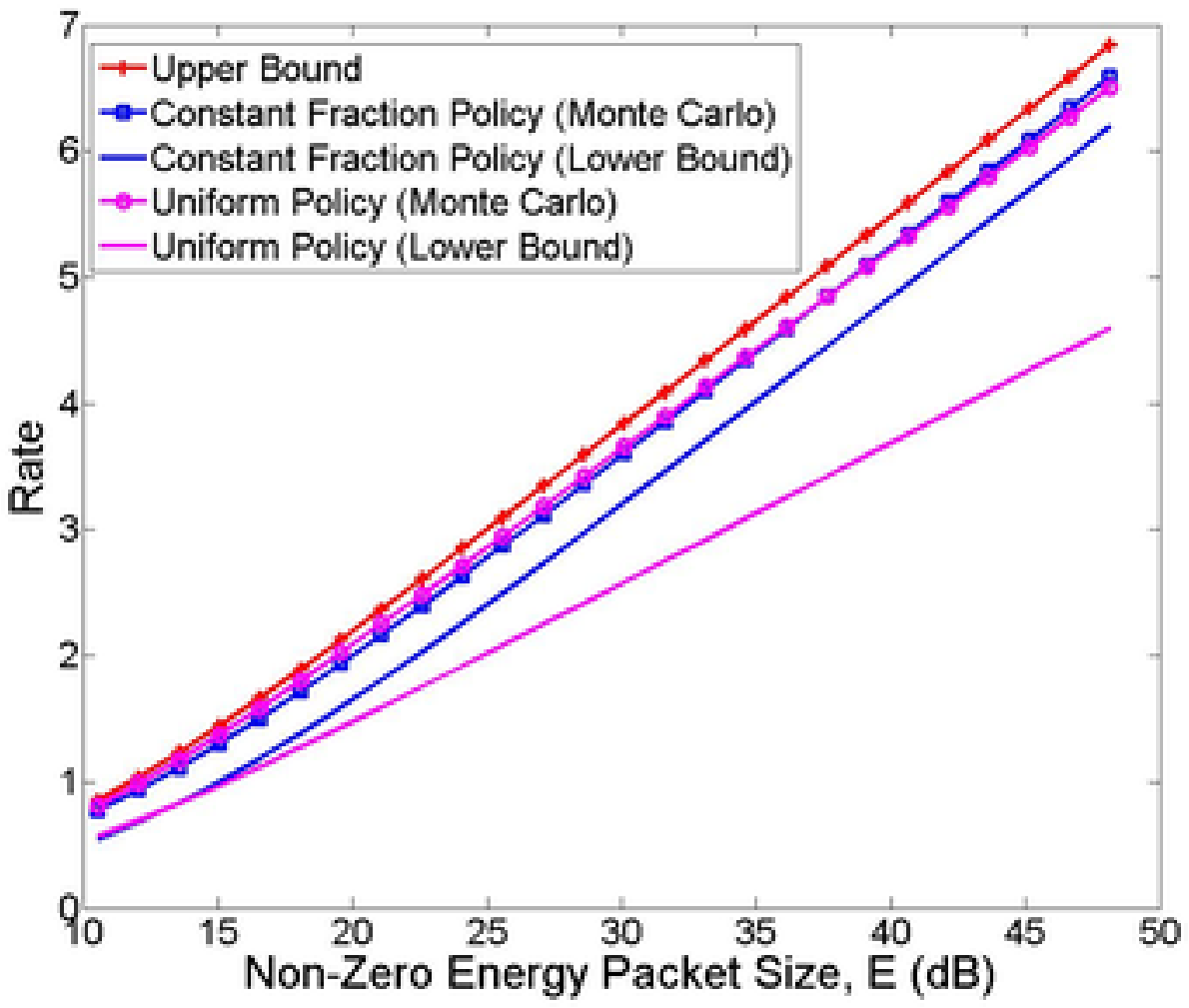}
     \label{fig: B_8E_5}
   }
 \caption{Numerical evaluations for a system described in Section \ref{sec: SysModel} with energy arrival probability $p = 1/5$.}
   \label{fig: p5}
\end{figure*}

In this section, we provide simulation results that compare the performance of the near optimal energy control policy we proposed in the earlier section, which we refer to as the {\it Constant Fraction Policy} in the current section, to some other simple energy/power control policies that have been proposed in the literature for the same problem. In particular, \cite{OzelTutuncuoglu_JSAC} proposes a {\it Constant Water Level Policy}, which we refer to as the {\it Uniform Policy} in this section, which allocates a constant amount of energy, equal to the average energy arrival rate, to each channel use as long as there is sufficient energy in the battery. When the energy in the battery is exhausted, no energy is allocated until the next energy arrival. \cite{SrivastavaKoksal} shows that this strategy becomes asymptotically optimal as  $B_{max}$ increases.\footnote{The strategy discussed in \cite{SrivastavaKoksal} is a slight variation of the {\it Uniform Policy} in that the amount of energy utilized in each channel use is  $\pm\epsilon$ of the average energy arrival rate. $\epsilon$ decreases to zero as $B_{max}\to\infty$ and the decay of $\epsilon$ controls the battery discharge probability and the convergence speed to the optimal average rate as $B_{max}\to\infty$. However, in its essence the strategy is a {\it Uniform Policy}, so its behavior is captured by the {\it Uniform Policy} plots.} The same intuition is suggested by \cite{OzelUlukus_BInfty} in the information-theoretic setting. 

Figure \ref{fig: p15} and \ref{fig: p5} summarize the results. In all the plots, the curve with label ``Upper Bound" is the upper bound on the average rate achieved by any feasible policy, i.e. $1/2 \log(1+ pB_{max})$ when $B_{max} \le E$ and $1/2 \log(1+ pE)$ when $B_{max} > E$. When $B_{max} \le E$, the curve with label ``Constant Fraction Policy" is the rate achieved by the strategy we proposed in the earlier section (Theorem~\ref{thm: const_gap_rate}), which uses a fixed $p$ fraction of the remaining energy in the battery; the curve with label ``Uniform Policy" is the rate achieved by a strategy that allocates $pB_{max}$ amount of energy, which is the average energy arrival rate in this case,  if there is enough energy left in the battery, and $0$ energy otherwise. Similarly when $B_{max} > E$, the {\it Constant Fraction Policy} uses a fixed fraction $p$ of the available energy in the battery $B_t$ and the {\it Uniform Policy} uses energy $pE$ at each channel use if possible. There are two curves for each of these strategies. The curves with ``(Lower Bound)" in the label represent an analytical lower bound we can compute on the rate achieved by these strategies by assuming $B_{max} = E$ (so these lower bounds remain the same as long as $B_{max} > E$). In particular, for the {\it Constant Fraction Policy}, the lower bound corresponds to the rate achieved by the strategy proposed in Theorem~\ref{thm: const_gap_rate2}. Recall the discussion in the paragraph after \eqref{eq:needeqnumber} which suggests a strategy using a constant fraction $p$ of the available energy would actually achieve a larger rate when $B_{max} > E$. This actual rate is obtained by running Monte Carlo simulation and is given by the curve with ``(Monte Carlo)" in the label. Similarly, it's possible to analytically compute a lower bound on the rate achieved by the {\it Uniform Policy} by assuming $B_{max} = E$ and the actual performance is obtained by Monte Carlo simulations.

Based on Fig. \ref{fig: B_le_E_15} and \ref{fig: B_le_E_5}, we see that in the $B_{max} \le E$ regime, the {\it Constant Fraction Policy} indeed tracks the upper bound within a constant gap for all values of $B_{max}$, whereas the gap of the {\it Uniform Policy} starts to diverge at around $15$ to $20$ dB depending on $p$. A similar conclusion holds for Fig. \ref{fig: B_2E_15} and \ref{fig: B_2E_5} where $B_{max} > E$ with the ratio of $B_{max}$ to $E$ being fixed at $2$. Fig. \ref{fig: B_8E_15} and \ref{fig: B_8E_5} showed that when $B_{max} \gg E$ (in this case $B_{max}= 8E$), the performances of both policies are very similar to each other across all SNR regimes. Moreover, they both track the upper bound very closely. This is not surprising given that  \cite{SrivastavaKoksal}, \cite{OzelUlukus_BInfty} show that the {\it Uniform Policy} converges to the upper bound as $B_{max}/E \to \infty$. Fig. \ref{fig: B_8E_15} and \ref{fig: B_8E_5} empirically show that our {\it Constant Fraction Policy} performs just as good when $B_{max} \gg E$. 

The figures also illustrate the difference between a constant gap guarantee on optimality and an asymptotic guarantee on optimality as $B_{max}\to\infty$. While Fig. \ref{fig: B_8E_15} and \ref{fig: B_8E_5} show that the {\it Uniform Policy} becomes optimal in the regime when $B_{max} \gg E$, when $B_{max}$ and $E$ are comparable this strategy can be arbitrarily away from the optimal rate (as illustrated in the figures for $B_{max} \le E$ and $B_{max}=2E$).\footnote{Indeed, this fact can be shown analytically; for example taking $B_{max}=E$, one can show that the gap between the upper bound and the rate achieved by the {\it Uniform Policy} increases to infinity as $B_{max}=E\rightarrow \infty$. However, we do not provide a proof of this fact since the trend is already quite obvious from the graphs.} Our {\it Constant Fraction Policy} on the other hand is within a bounded gap from optimality for all parameter values as guaranteed by our theoretical results. Finally, observe that there is not much qualitative difference between Fig. \ref{fig: p15} and  Fig. \ref{fig: p5} which correspond to  different values of energy arrival probability $p$. This is expected as our constant gap guarantees hold independent of $p$.

\section{The information theoretic capacity of the finite battery system}
\label{sec: Info_Cap}
 
In this section, we approach the system in Section~\ref{sec: SysModel} with Bernoulli energy arrival process from an information theoretic perspective. In particular, we derive an upper and a lower bound on the information-theoretic capacity of the channel and show that the gap between these upper and lower bounds is no more than a constant for all choices of the system parameters. Our lower bound relies on the near-optimal energy allocation policy we developed in Section~\ref{sec: Online_Policy}, and reveals the connection between the two problems by developing a codebook construction which allows to implement a given energy allocation policy. We examine the regime where $B_{max} \le E$ in Section \ref{subsec: B_less_E}, and the regime $B_{max} > E$ in Section \ref{subsec: B_greater_E}. 
 
\subsection{The $B_{max} \le E$ regime}
\label{subsec: B_less_E}

In the case $B_{max} \le E$, each time the non-zero energy packet arrives, the battery will be filled up completely regardless of how much energy was remaining in the battery. In particular, at least $E - B_{max}$ amount of energy is wasted in every non-zero incoming energy packet. Therefore, a system with energy packet size $E' = E - (E- B_{max}) = B_{max}$ is equivalent to the original system in terms of the available communication resources and hence the two systems must have the same capacity. In the sequel, we consider the equivalent system with $E' = B_{max}$. 

Note that as a result of the above observation, in the regime $B_{max} \le E$, the capacity of the system can only explicitly depend on $B_{max}$ and $p$ and not $E$. We next provide an upper bound on the capacity as a function of $B_{max}$ and $p$. We will then provide an achievable scheme and show that the rate it achieves is within a constant gap from this upper bound for all choices of $B_{max}$ and $p$. 

\begin{thm}[Upper Bound on Capacity: $B_{max} \le E$]
\label{thm: UpperBoundC}
When $E \ge B_{max}$, the capacity $C$ of the channel with Bernoulli energy arrival process defined in Section \ref{sec: SysModel} is upper bounded by 
\begin{equation}
\label{eq: CapUpperBound}
C \le \frac{1}{2} \log \left(1 + p B_{max}\right) \triangleq C_{ub}( B_{max}, p). 
\end{equation} 
\end{thm}
\begin{proof}
Note that the capacity of the channel in Section~\ref{sec: SysModel}  should be an increasing function of the battery size,  since we can always choose not to use the extra battery space. Therefore, the capacity of our channel with finite battery is upper bounded by the capacity of the same channel with infinite battery size. The infinite battery capacity has been characterized in \cite{OzelUlukus_BInfty} as
\begin{align}
 C_{\infty} & = \frac{1}{2}\log(1 + \mathbb{E}\left[ E_t \right] ) \notag \\
 & = \frac{1}{2}\log(1 + pE) \label{eq: ClassicAWGN}
\end{align}
since $\mathbb{E}\left[ E_t \right] = pE$ in our current case. Based on the earlier discussion, when $E \ge B_{max}$, the capacity of the system is the same as the capacity of a system with reduced energy packet size $E= B_{max}$. Plugging $E= B_{max}$ in \eqref{eq: ClassicAWGN}  gives the desired upper bound in \eqref{eq: CapUpperBound}.
\end{proof}

Next, we will provide an achievable strategy for the channel in Fig.~\ref{fig: model} when the energy arrival process $\{E_t\}$ is also causally known at the receiver. Later, we will use this result to derive an achievable rate for our original system in Section~\ref{sec: SysModel} where we assume that the receiver has no information about the energy arrival process.

\begin{thm}[Achievable Scheme with CSIR]
\label{thm: Achievable}
Assume that for the system defined in Section \ref{sec: SysModel}, the Bernoulli energy arrival process $\{E_t\}$ is causally known not only at the transmitter but also at the receiver. Then we can achieve any rate:
\begin{equation}
R_{ach} \leq  \displaystyle\sum\limits_{j = 0}^{\infty} p (1-p)^j \max_{p(x): |X|^2 \le \mathcal{E}_j} I(X; Y) \label{eq: Achievable} 
\end{equation}
for any non-negative $\mathcal{E}_0, \mathcal{E}_1, ...$ satisfying
\begin{equation}
\sum\limits_{j = 0}^{\infty} \mathcal{E}_j \le B_{max}. \label{eq: EnergyConstr}
\end{equation}
\end{thm}
The proof of Theorem \ref{thm: Achievable} is provided in the Appendix.

The idea for the achievable scheme is that if both the transmitter and receiver know when the energy packet $E$ arrives, they can agree on an energy allocation strategy ahead of time. As we did in Section~\ref{sec: Online_Policy}, here we concentrate on an energy allocation policy $\mathcal{E}_j$ that is invariant across different epochs (the period of time between two adjacent non-zero packet arrivals). $\mathcal{E}_j$ denotes the amount of energy allocated to transmission, $j$ channel uses after the last time the battery was recharged, i.e. if energy arrives at the current channel use, we allocate $\mathcal{E}_0$ amount of energy for transmission; if energy arrived in the previous channel use but not the current channel use, then we allocate $\mathcal{E}_1$ amount of energy for transmission, etc. The transmitter and receiver agree ahead of time on a sequence of codebooks $\mathcal{C}^{(j)}$ where each codebook is amplitude-constrained to $\mathcal{E}_j$, i.e. the symbols of each codeword in $\mathcal{C}^{(j)}$ are such that $|X|^2 \le \mathcal{E}_j$. This ensures that the symbol transmitted at the corresponding time will not exceed the energy constraint $\mathcal{E}_j$. We assume that the transmitter has one codeword $c_j\in \mathcal{C}^{(j)}$ from each codebook to communicate to the receiver and the symbols of these codewords are interleaved as dictated by the realization of the energy arrival process. For example, upon the arrival of the first energy packet, the transmitter sends the first symbol of codeword $c_0$; if there is no energy packet arrival in the next channel use, it transmits the first symbol of codeword $c_1$ in the next channel use, etc. Once the second energy packet arrives, the cycle is reset and the transmitter moves to transmitting the second symbol of the codeword $c_0$, then the second symbol in codeword $c_1$, etc. (See the Appendix for a detailed description of the strategy and its performance analysis.) \eqref{eq: Achievable} gives the rate we can achieve with such a strategy in the large blocklength limit. \eqref{eq: EnergyConstr} ensures that the total energy spent does not exceed the available energy in the battery at $j=0$ which is equal to $B_{max}$ (since when $E\geq B_{max}$, the battery is recharged to full every time an energy packet arrives).

We next show that when there is no channel state information at the receiver we can achieve the rate in Theorem~\ref{thm: Achievable} with at most $H(p)$ bits of penalty. The main idea of the result is similar to Theorem 1 in \cite{Jafar_CSIR} which shows that over an information stable channel the maximum possible capacity improvement due to receiver side information is bounded by the amount of the side information itself. In the current case, the energy constraints in \eqref{eq: ampConst} and \eqref{eq: batteryUpdate} introduce memory into the system and its not a priori clear if the channel is information stable or not. The following theorem extends Theorem 1 of \cite{Jafar_CSIR} to general, not necessarily information-stable, channels. 
\begin{thm}[Capacity improvement due to RX Side Info]\label{thm:CapCSIR} Consider a general channel, not necessarily stationary memoryless, defined as a sequence  $\{ W^n(\cdot|\cdot)=P_{Y^{(n)}|X^{(n)}}:\mathcal{X}^{(n)}\rightarrow \mathcal{Y}^{(n)}\}_{n=1}^\infty, $ of arbitrary  conditional probability distributions   together with an input and an output alphabet for each $n$ (which need not be Cartesian products of a basic input and an output alphabet). The improvement in channel capacity due to the availability of side information at the receiver is upper-bounded by the \textit{spectral sup-entropy rate} of the side information process $G=\{G^n\}_{n=1}^\infty$, which is defined as 
\begin{equation*}
\bar{H}(G)=\text{p-}\limsup_{n\rightarrow\infty}\frac{1}{n}\log \frac{1}{P_{G^n}(g^n)},
\end{equation*}
where $\text{p-}\limsup$ denotes \textit{limsup in probability} (see \cite[p.14]{Han}).
\end{thm}

The proof of the theorem is given in the Appendix. In order to apply the theorem to our current channel with causal transmitter side information, we can use Shannon's technique in \cite{Shannon} to first transform the channel to an equivalent channel without states but with an enlarged input alphabet over so called Shannon strategies (this transformation has been developed in \cite{MaoHassibi}) and then apply Theorem~\ref{thm:CapCSIR} to the equivalent channel. Since the side information process $G$ in our case is the i.i.d. Bernoulli energy arrival process $\{E_t\}$, which has entropy rate $\bar{H}(G)=H(p)$ we immediately get the following theorem.

\begin{thm}[Achievable Rate without CSIR]
\label{thm: Achievable2}
Consider the system defined in Section \ref{sec: SysModel} where the energy arrival process $\{E_t\}$ is causally known only at the transmitter, but not at the receiver. The capacity $C$ of this system is lower bounded by
\begin{align}
C \geq  \displaystyle\sum\limits_{j = 0}^{\infty} & \; p  (1-p)^j \max_{p(x): |X|^2 \le \mathcal{E}_j} I(X; Y)- H(p) \label{eq: Achievable2}  \\
& \text{subject to: } \;\;\; \displaystyle\sum\limits_{j = 0}^{\infty} \mathcal{E}_j \le B_{max} \label{eq: EnergyConstr2}
\end{align}
where $H(p)$ is the binary entropy function.
\end{thm}

While the proof of Theorem~\ref{thm: Achievable2} follows directly from Theorem~\ref{thm:CapCSIR}, in the Appendix we provide an alternative proof for Theorem~\ref{thm: Achievable2} that makes specific use of our channel structure and the achievable scheme we propose in Theorem~\ref{thm: Achievable}.

Note that the mutual information maximization problem in \eqref{eq: Achievable2} is over the class of distributions with bounded support in $\left[ - \sqrt{\mathcal{E}_j}, \sqrt{\mathcal{E}_j} \right]$. This is a nontrivial optimization problem and as shown in \cite{Smith}, the optimal input distribution turns out to be discrete rather than continuous. Below we lower bound the optimal value of this optimization problem by considering the uniform distribution over the interval $\left[ - \sqrt{\mathcal{E}_j}, \sqrt{\mathcal{E}_j} \right]$.
 
\begin{lem}[LB on Amplitude-constrained AWGN Capacity]
\label{lem: Gap_Uniform}
\begin{equation}
\label{eq: UniformCap}
\max_{p(x): |X|^2 \le \mathcal{E}_j} I(X; Y) \ge \frac{1}{2} \log \left(1 + \mathcal{E}_j  \right) - 1.04
\end{equation}
\end{lem}

\begin{proof}
Take $p(x)$ to be the uniform distribution on $\left[ - \sqrt{\mathcal{E}_j}, \sqrt{\mathcal{E}_j} \right]$, then the average power is $\mathbb{E}[X^2] = \frac{\mathcal{E}_j}{3}$. We use the following result which is proved in \cite[Eq.(7)]{OzarowWyner}:
\begin{prop} Consider a discrete-time AWGN channel with average power constraint $P$, i.e. $\mathbb{E}[|X|^2]\leq P$, and noise variance $\sigma^2$. Let the input $X$ of the channel be uniformly distributed over $[-\sqrt{3P},\sqrt{3P}]$ (so that $\mathbb{E}[|X|^2]=P$) and $Y$ be the corresponding output random variable. Then 
$$
I(X;Y)\geq \log\left(1+\frac{P}{\sigma^2}\right)- \frac{1}{2}\log\left(\frac{\pi e}{6}\right).
$$
\end{prop}
By using the result in the proposition, we get
\begin{align}
\max_{p(x): |X|^2 \le \mathcal{E}_j} & I( X;  Y) \ge \frac{1}{2} \log \left(1 + \frac{\mathcal{E}_j}{3} \right)  - \frac{1}{2}\log\left(\frac{\pi e}{6}\right) \notag\\
 & \ge \frac{1}{2} \log(1 + \mathcal{E}_j)  - \left(\frac{1}{2}\log(3) +  \frac{1}{2}\log\left(\frac{\pi e}{6}\right) \right) \notag \\
 & = \frac{1}{2} \log(1 + \mathcal{E}_j) - 1.04
\end{align}
\end{proof}

Combining Lemma \ref{lem: Gap_Uniform} and Theorem \ref{thm: Achievable2}, we can conclude that the following rate is achievable in the system defined in Section \ref{sec: SysModel} where the energy arrival process is causally known only at the transmitter, but not at the receiver:
\begin{align}
\label{eq: achievable_2}
R_{ach} \ge  \displaystyle\sum\limits_{j = 0}^{\infty} p & (1-p)^j \frac{1}{2} \log \left(1 + \mathcal{E}_j \right) - H(p) - 1.04 
\end{align} 
with $\mathcal{E}_{j}$ subject to the constraint \eqref{eq: EnergyConstr2}. Notice that up to a fixed constant, we are back to the online rate optimization problem studied in Section \ref{sec: Online_Policy} (see eq.\eqref{eq: ach}). Employing the near-optimal allocation policy from the previous section, which assigns $\mathcal{E}_j = p(1-p)^jB_{max}$, gives us the following achievable rate. 

\begin{lem}[Lower Bound on Capacity: $B_{max} \le E$]
\label{lem: LowerBoundC}
When $E \ge B_{max}$, the capacity of a system defined in Section \ref{sec: SysModel} with Bernoulli energy arrival process $\{E_t\}$ only causally known at the transmitter, but not at the receiver,  is lower bounded by
\begin{align}
C& (B_{max}, p)  \ge C_{lb}(B_{max}, p) \notag\\
& \triangleq  \left( \displaystyle\sum\limits_{j = 0}^{\infty} p(1-p)^j \frac{1}{2} \log\left( 1 + (1-p)^j p B_{max}  \right) - K(p)\right)^+ \label{eq: CapLowerBound}
\end{align}
where $a^+=\max(a, 0)$ and $K(p) = 1.04 + H(p)$ where $H(p)$ is the binary entropy function. 
\end{lem}

\begin{proof}
This lemma is a direct consequence of Theorem \ref{thm: Achievable2} and Lemma \ref{lem: Gap_Uniform} with particular choice of $\mathcal{E}_j = p(1-p)^jB_{max}$. We complete the proof by verifying $\{\mathcal{E}_j\}$ satisfies the constraint \eqref{eq: EnergyConstr2} since $\displaystyle\sum\limits_{j = 0}^{\infty} p (1-p)^j B_{max} = B_{max}$, and use the fact that the capacity is non-negative. 

\end{proof}

The next theorem states that for all values of $B_{max}$ and $p$, the gap between the lower bound in \eqref{eq: CapLowerBound} and the upper bound in \eqref{eq: CapUpperBound} is bounded by a constant. 

\begin{thm} [Constant Gap]
\label{thm: ConstantGapBlessE}
For all $0<p<1, \,B_{max} >0$, we have the following inequality:
\begin{align}
\label{eq: gap}
C_{ub}(B_{max}, p) - C_{lb}(B_{max}, p) \le 2.58 \mbox{ bits}
\end{align}
\end{thm}

The detailed proof is given in the Appendix. Notice that in the proof of Theorem \ref{thm: const_gap_rate}, we have already bounded the gap between $\displaystyle\sum\limits_{j = 0}^{\infty} p(1-p)^j \frac{1}{2} \log\left( 1 + (1-p)^j p B_{max}  \right)$ and $\frac{1}{2}\log\left(1+pB_{max} \right)$ by a constant, so it's no surprise we have a constant bound here. The additional term $K(p)$ in Lemma \ref{lem: LowerBoundC} is why we have a larger constant, i.e. 2.58 bits, here. 

\subsection{The $B_{max} > E$ regime}
\label{subsec: B_greater_E}

In Section \ref{subsec: B_less_E}, we provided an upper bound and a lower bound on the capacity when $B_{max} \le E$. Because the energy must be stored into the battery before it can be used, the extra energy is wasted when $B_{max} \leq E$. Therefore, the  capacity depends only on $B_{max}$ and not $E$ in that case. When $B_{max} > E$, a natural question is how the extra battery space impacts the capacity of the system, and whether $B_{max}$ or $E$ (or both) is the determining factor for the capacity. We show below that the capacity depends critically on $E$ in this case and not so much on the extra battery space $B_{max}$.

The bounds in Section~\ref{subsec: B_less_E} can be summarized as:
\begin{equation}
\label{eq: two_sided}
\frac{1}{2} \log \left(1 + p B_{max} \right) - 2.58 \le C \le \frac{1}{2} \log \left(1 + p B_{max} \right)
\end{equation}
Notice, in the case $B_{max} = E$, this bound turns into:
\begin{equation}
\label{eq: B_equal_E}
\frac{1}{2} \log \left(1 + p E \right) - 2.58 \le C(B_{max} = E, p) \le \frac{1}{2} \log \left(1 + p E \right)
\end{equation}

Note that the right hand side of \eqref{eq: B_equal_E} is the same as the infinite battery capacity given by \eqref{eq: ClassicAWGN}. This gives rise to the following theorem. 

\begin{thm} [Bounds on Capacity: $B_{max} > E$]
\label{thm: B_greater_E}
When the battery size $B_{max} > E$, the capacity of a system with Bernoulli energy arrival process defined in Section \ref{sec: SysModel} is bounded by:
 \begin{equation}
\label{eq: B_greater_E}
\frac{1}{2} \log \left(1 + p E \right) - 2.58 \le C \le \frac{1}{2} \log \left(1 + p E \right)
\end{equation}
\end{thm}

\begin{proof}
Since having a larger battery can only help the capacity of the system, the capacity when $B_{max}>E$ is at least as large as the capacity when $B_{max}=E$. Using the left hand side of  \eqref{eq: B_equal_E}, we have: when $B_{max}>E$
\begin{align*}
C  & \ge C(B_{max} = E, p) \\
& \ge \frac{1}{2} \log \left(1 + p E \right) - 2.58
\end{align*}
Similarly, using \eqref{eq: ClassicAWGN}, we have the upper bound:
\begin{align*}
C & \le C(B_{max} = \infty, p) \\
& \leq \frac{1}{2} \log \left(1 + p E \right) 
\end{align*}
This completes the proof of the theorem.
\end{proof}

\section{Generalization to other energy profiles}
\label{sec: General_EP}
The near optimal energy allocation policy and the coding strategy we developed in the earlier sections, as well as the resultant constant gap approximation for the information-theoretic capacity of the AWGN energy harvesting communication channel were specific to the i.i.d. Bernoulli energy arrival process and it may seem difficult to extend these results to the more general settings. In this section, we present a simple way to apply these strategies in the more general settings of  i.i.d. energy arrival processes.  

The idea we propose is very simple: given an i.i.d. energy arrival process where $E_t$ has an arbitrary distribution (discrete or continuous), fix an energy level $E$ and find the probability $p$ of having an energy arrival with packet size at least $E$, i.e. $p=\mathbb{P}(E_t\geq E)$; then apply the near optimal energy allocation policy of Section~\ref{sec: Online_Policy} or the near optimal coding strategy of Section~\ref{sec: Info_Cap} as if the exogenous energy process were i.i.d. Bernoulli with parameters $E$ and $p$. Clearly, this is an energy feasible strategy for the general i.i.d. energy arrival process. However, a priori it may seem highly suboptimal and wasteful of energy since by treating the energy arrival process as Bernoulli with parameters $E$ and $p$, we assume that there is no energy arrival when the arriving energy packet size is smaller than $E$, and we assume that the energy packet size is equal to $E$ each time we receive an energy packet of size larger than or equal to $E$. Effectively, our strategy may not be utilizing a large fraction of the energy accumulating in the battery under the general i.i.d. energy arrival process. However, as we illustrate next this simple strategy turns out to be sufficient to achieve the optimal value for the energy allocation problem in Section~\ref{sec: Online_Policy} and the information-theoretic capacity of the channel within a constant gap for a large class of i.i.d. energy arrival processes, though with a larger constant. Note that as already discussed in Section \ref{sec: simulations}, the earlier strategies proposed in the literature \cite{OzelTutuncuoglu_JSAC,MaoKoksalShroff,SrivastavaKoksal} would fail to provide a constant gap approximation.

We have the following theorem  which is the extension of our main result in Theorem~\ref{thm: MainResult} to the general i.i.d. energy arrival processes. It is straightforward to write down the corresponding extensions of Theorems~\ref{thm: const_gap_rate} and \ref{thm: const_gap_rate2} to this more general setting.

\begin{thm}
\label{thm: Extension}
Assume that the energy arrival process $\lbrace E_t\rbrace$ is an i.i.d. random process with each $E_t$ distributed according to an arbitrary cumulative distribution function $F(x)$ such that $F(x)=0, \forall x<0$. Then for each $0\leq x\leq B_{max}$, the information theoretic capacity $C$ of the corresponding AWGN energy harvesting channel with battery size $B_{max}$ is bounded by
\begin{equation}
\label{eq: mainres_general}
\begin{split}
\frac{1}{2}\log(1+ & x (1-F(x)))  - 2.58 \leq C
\\ &\leq \frac{1}{2}\log\left(1+\int_0^{B_{max}} \! (1-F(y) ) \, \mathrm{d}y \right).
\end{split}
\end{equation}
\end{thm}
\begin{proof}
Fix an energy level $x$. The capacity $C$ of our channel can be lower bounded by the capacity of the same channel when the exogenous energy arrival process is i.i.d. Bernoulli with energy  packet size $E=x$ and probability of energy packet arrival $p=1-F(x)$ since we can always discard the additional energy. Then, applying Theorem \ref{thm: MainResult} we immediately get the lower bound in \eqref{eq: mainres_general}. 

For the upper bound, we notice when the arrival energy $E_t > B_{max}$, the extra $E_t - B_{max}$ amount of energy is wasted. Therefore, the available resource for communication for this system is equivalent to a system where $E_t$, is distributed according to a cumulative distribution function:
\begin{align}
\label{eq: equiv_cdf}
\tilde{F}(x) = \left\{ \begin{array}{rl}
0 & \mbox{for $x < 0$} \\
F(x) & \mbox{for $0 \le x < B_{max}$ } \\
1 & \mbox{for $B_{max} \le x$}
\end{array}  \right.
\end{align}

Again, using the infinite battery capacity, we can upper bound the capacity of the system by:
\begin{align*}
C  & \le C_{\infty} \\
 & = \frac{1}{2}\log(1+\mathbb{E}(E_t)) \\
 & \stackrel{(a)}{=} \frac{1}{2} \log\left(1+\int_{0}^{\infty}(1-\tilde{F}(y)) \mathrm{d}y \right) \\
 & \stackrel{(b)}{=} \frac{1}{2} \log\left(1+\int_{0}^{B_{max}}(1-F(y)) \mathrm{d}y \right)
\end{align*}
where (a) comes from a well-known identity relating the cdf of a non-negative random variable to its mean, (b) uses \eqref{eq: equiv_cdf}.
\end{proof}

Note that the above theorem holds for any value of $x$ between $0$ and $B_{max}$. For a given energy arrival process, we would want to optimize $x$ to maximize the lower bound on capacity. In other words, the lower bound in the theorem can be equivalently rewritten as
\begin{equation}
\label{eq: gen_lower_bound}
{\frac{1}{2}\log\left( 1+ \sup_{0\leq x \leq B_{max}}{x (1-F(x))}\right) - 2.58} \leq C .
\end{equation}
Using this lower bound, the approximation gap in the theorem (the difference between the lower and upper bounds) is given by
\begin{equation}
\label{eq: gen_gap}
\frac{1}{2}\log\left( \frac{1+\int_{0}^{B_{max}}(1-F(y)) \mathrm{d}y}{1+ \displaystyle\sup\limits_{0\leq x \leq B_{max}}{x (1-F(x))}} \right) +2.58,
\end{equation}   
which can be upper bounded by
\begin{equation}
\label{eq: ratio}
\frac{1}{2}\log\left( \frac{\int_{0}^{B_{max}}(1-F(y)) \mathrm{d}y}{\displaystyle\sup\limits_{0\leq x \leq B_{max}}{x (1-F(x))}} \right) +2.58.
\end{equation}
since the numerator is always greater than or equal to the denominator inside the log in Equation \eqref{eq: ratio}. 

Figure \ref{fig: one-CDF} illustrates the ratio inside the logarithm for a given energy arrival process, characterized by the graph $1-F(x)$ between $0$ and $B_{max}$. The numerator in this fraction is the area under this graph, while the denominator is the largest area of a rectangle lying below this graph (as illustrated by the shaded rectangle). Therefore, as long as the cumulative distribution function has the property that this ratio is not too large, Theorem~\ref{thm: Extension} will yield a constant gap approximation of the capacity. As we illustrate next, two examples of such distributions are the uniform distribution and an energy arrival process with two discrete energy levels (as opposed to one energy level in the Bernoulli case). For these distributions we get a total approximation gap of $3.08$  bits (as opposed to $2.58$ bits in the Bernoulli case). However, not all cumulative distribution functions yield a constant gap approximation. We also provide a counter example (a sequence of energy profiles) for which the approximation gap provided by the theorem becomes arbitrarily large.

\begin{figure}[t!]
\centering
\scalebox{0.5}
{\includegraphics{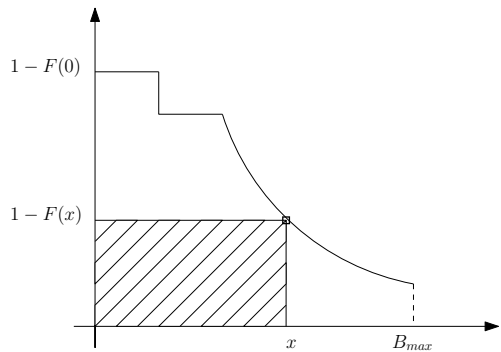}}
\caption{The graph of $1-F(x)$ for an arbitrary distribution. The approximation gap for the capacity is related to the ratio between the area under this graph and the area of the largest rectangle lying below the graph.}
\label{fig: one-CDF}
\end{figure}

\begin{figure}[t!]
\centering
\scalebox{0.53}
{\includegraphics{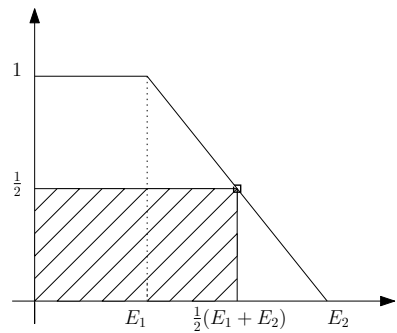}}
\caption{The graph of $1-F(x)$ for a uniform distribution. Observe that choose $x=(A_1+A_2)/2$ ensures the ratio of the area under the graph to the corresponding rectangle's area to be $2$ for all $A_1, A_2$.}
\label{fig: Uniform-CDF}
\end{figure}

\subsection{Uniform Distribution}
Assume the energy arrival process is i.i.d. with $E_t$ uniformly distributed over the interval $[A_1,A_2]$ for some arbitrary $0 \le A_1 < A_2$. We first assume that $B_{max} \geq A_2$. We have
\begin{equation}
\begin{split}
\int_{0}^{B_{max}}(1-F(y)) \mathrm{d}y &= \frac{A_1 +A_2}{2}
\\ & = 2\times\frac{1}{2}\times\frac{A_1 +A_2}{2}
\\ &\leq 2\times {\sup_{0\leq x \leq B_{max}}{x (1-F(x))}}.
\end{split}
\end{equation}
we have the inequality in the last step because by choosing $x=(A_1 +A_2)/2$, we can achieve ${x (1-F(x))}=(A_1 +A_2)/4$ (See Figure \ref{fig: Uniform-CDF}).  With this choice of energy level $E$, i.e. the midpoint of the interval $[A_1, A_2]$, the ratio inside the log in Equation \eqref{eq: ratio} is guaranteed to be a constant of $2$ for any i.i.d. uniformly distributed energy arrival process independent of the values of $A_1$ and $A_2$.
In particular, Theorem \ref{thm: Extension} will approximate the capacity as 
\begin{equation}
\label{eq: upper_uniform}
\frac{1}{2}\log\left(1+ \frac{A_1+A_2}{2}\right)\quad\text{bits}
\end{equation}
when $B_{max} \geq A_2$ within a gap of $3.08$ bits.

When $(A_1+ A_2)/2 \leq B_{max}< A_2$, we can again choose $x =(A_1 + A_2)/2$, and achieve within 3.08 bits of \eqref{eq: upper_uniform}, which clearly is also an upper bound for this case. When $A_1 \le B_{max}  < (A_1+A_2)/2$, we can no longer choose $x = (A_1 + A_2)/2$, since $x$ must be no more than $B_{max}$. However, in this case, if we simply choose $x = B_{max}$, the ratio of interest is still no more than $2$, therefore, Theorem \ref{thm: Extension} will approximate the capacity as $1/2 \log(1 + B_{max})$ within 3.08 bits. Finally, when $B_{max} \le A_1$, we are back to a degenerate Bernoulli case with arrival probability $p=1$, and packet size $B_{max}$, so we have upper bound $1/2 \log(1 + B_{max})$, and lower bound of $1/2 \log(1 + B_{max}) - 2.58$.

We summarize the approximate capacity and the approximation gap  for various regimes in the  following table:

\begin{table}[H]
\begin{center}
\caption{Approximate Capacity for AWGN Energy Harvesting Channel with Uniform Energy Arrival Distribution}
\begin{tabular}{|c|c|c|}
\hline
Regime & Approximate Capacity & Gap \\
\hline\hline
$\frac{A_1 + A_2}{2} \le B_{max}$ & $\frac{1}{2} \log \left(1 + \frac{A_1+A_2}{2} \right)$ & $\le 3.08$ \\
\hline
$A_1 \le B_{max} < \frac{A_1 + A_2}{2}$ & $\frac{1}{2}\log\left(1 + B_{max} \right)$ & $\le 3.08$\\
\hline 
$B_{max} < A_1$ & $\frac{1}{2} \log \left(1 + B_{max}\right)$  & $\le 2.58$ \\
\hline
\end{tabular}
\label{table: Uniform}
\end{center}
\end{table}

Similar to the Bernoulli arrival process, based on Table \ref{table: Uniform}, we infer there are two qualitatively different regimes for the capacity of a system with uniform arrival process. In particular, when $B_{max} < (A_1 + A_2)/2$, the capacity is increasing roughly logarithmic in $B_{max}$, while when $B_{max} \geq (A_1 + A_2)/2$ the capacity approximately saturates to $\frac{1}{2} \log \left(1 + \frac{A_1+ A_2}{2} \right)$. Perhaps, what's interesting here is that the threshold where this regime shift happens is neither $A_1$ nor $A_2$ by itself, but instead happens at the midpoint between $A_1$ and $A_2$.   

\subsection{k-Level Distribution}\label{subsec: CE}
Assume now that the energy arrival process is again i.i.d. but $E_t$ is a discrete random variable that takes the value $A_i$ with probability $p_i$  for $i=1,2,\ldots ,k$. Further, assume that $0 < A_1 < A_2 < \cdots < A_k \le B_{max}$. Note that the probability of having no energy packets is $1-\sum_{i=1}^{k}{p_i}$. Then we have
\begin{align*}
\int_{0}^{B_{max}}(1-F(y)) \mathrm{d}y &= \sum_{i=1}^{k}{p_i A_i}
\\ &\le \sum_{i=1}^{k}{A_i \sum_{j=i}^{k}{p_j}}
\\ &\le k\max_{1\leq i\leq k}{A_i \sum_{j=i}^{k}{p_j}}
\\ &= k {\sup_{0\leq x \leq B_{max}}{x (1-F(x))}}.
\end{align*}
The last equality holds because the supremum is achieved when $x$ in Theorem~\ref{thm: Extension} approaches $A_i$ from left for which ${A_i \sum_{j=i}^{k}{p_j}}$ is maximized. This shows we can approximate the capacity as $\frac{1}{2}\log(1+\mathbb{E}[E_t])$, when $0 < A_i \le B_{max}$ for all $i$, within a gap of at most 
\begin{equation}
\label{eq: k_level_gap}
\frac{1}{2}\log(k)+2.58\quad\text{bits}.
\end{equation}

For example, for $k=1$, i.e. the Bernoulli case, we recover the gap of $2.58$ bits. For the case $k=2$, we obtain an approximation gap of $3.08$ bits. Note that Equation \eqref{eq: k_level_gap} is just an upper bound on the gap, and not necessarily tight. In particular, not all classes of distributions have gap increasing with k. For example, if all the $p_i$'s are the same and $A_i$'s are equally spaced within an interval $[A_a, A_b]$. Then as $k \to \infty$, the distribution of $E_t$ starts to approach a uniform distribution we examined in the previous subsection, which we know has a bounded gap of $3.08$ bits. However, there does exist sequence of profiles, where as $k$ increases, our approximation gap increases unboundedly. Indeed, in the last subsection, we will find an example of discrete profiles for which the gap can be arbitrarily large.

If $A_{i-1}\leq B_{max} < A_i$, we can consider the equivalent distribution where we replace all $A_j$, $j\geq i$ with $B_{max}$ and assign a probability mass $\sum_{j=i}^{k}{p_j}$ to $B_{max}$. This case can be handled similarly as we did above but note that in this case the expression for the approximate capacity will depend on $A_j$ and $p_j$'s for $j < i$ as well as $B_{max}$.

\subsection{Counterexample}
While the strategy has worked well for the above cases, here we provide a sequence of profiles for which the gap between the upper and lower bounds in Theorem \ref{thm: Extension} can be made arbitrarily large. Define a sequence of profiles with cdf $\{F_n(x)\}_{n=1}^{\infty}$ as
\begin{align*}
1-F_n(x) = \left\{ \begin{array}{rl}
1 & \mbox{$x < 1$} \\
\frac{1}{x} & \mbox{$1\leq x < n$} \\
0 & \mbox{$n \le x$}
\end{array}  \right.
\end{align*}
Now, suppose we have a channel with energy arrival process described by the cdf $F_n(x)$. Without loss of generality, let $B_{max} = n$, i.e. the process takes on a support in $[0, B_{max}]$. Then,
\begin{align}
\frac{1+\int_{0}^{B_{max}}(1-F_n(y)) \mathrm{d}y}{1+{\displaystyle\sup\limits_{0\leq x \leq B_{max}}{x (1-F_n(x))}}} & \stackrel{(a)}{=} \frac{1+\int_{0}^{1}1 \mathrm{d}y +\int_{1}^{n}\frac{1}{y} \mathrm{d}y}{1+1} \notag \\ 
&= 1+\frac{\ln(n)}{2} \label{eq: unbounded_ratio}
\end{align}
where (a) is true due to the way we constructed $F_n(x)$, in particular:
\begin{align*}
x(1-F_n(x)) = \left\{ \begin{array}{rl}
x & \mbox{for $x < 1$} \\
1 & \mbox{for $1 \le x < n$ } \\
0 & \mbox{for $n \le x $}
\end{array}  \right.
\end{align*}
Substitute the ratio from \eqref{eq: unbounded_ratio} into Equation \eqref{eq: gen_gap}, we obtain a gap between the upper and lower bounds in Theorem \ref{thm: Extension} of:
\begin{equation}
\label{eq: unbounded_gap}
\frac{1}{2}\log\left(1 + \frac{\ln(n)}{2} \right) + 2.58
\end{equation}
which can be made arbitrarily large as $n$ grows unboundedly. This suggests that the approximation gap using our strategy can not be bounded by a constant as we did in the previous examples.  Alternatively, the unboundedness of the ratio can be seen from Figure \ref{fig: Counterexample} where the area of every shaded bounded rectangle is no more than $1$, but the area under the graph becomes arbitrarily large as $n$ goes to infinity. 

\begin{figure}[t!]
\centering
\scalebox{0.5}
{\includegraphics{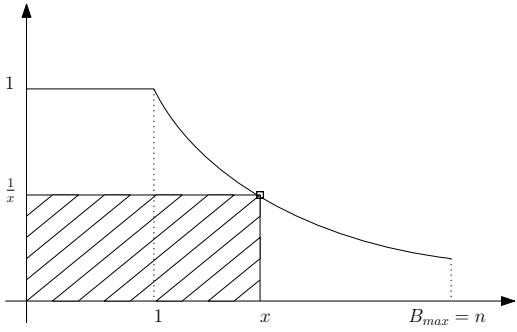}}
\caption{The graph of $1-F_n(x)$. Note that the area under the graph goes to infinity as $n$ goes to infinity, but the hatched area $x(1-F(x))$ is no more than $1$ for every $x$.}
\label{fig: Counterexample}
\end{figure}

As mentioned earlier, it is possible to  come up with a discrete counterexample by discretizing this sequence of profiles at integer levels. In particular, let the energy arrival process be a $k$-level distribution with $A_i= i$ for $1 \le i \le k$, $p_i=\frac{1}{i}-\frac{1}{i+1}$ for $1 \le i \le k-1$ and $p_k = \frac{1}{k}$. Again, let $B_{max}=k$ to avoid complication due to truncation of the cdf. Then the ratio inside the log in \eqref{eq: gen_gap} can be written as:
\begin{align*}
\frac{1+\int_{0}^{B_{max}}1-F_k(y) \mathrm{d}y}{1+{\displaystyle\max\limits_{1\leq i \leq k}{A_i (1-F_k(A_i^-))}}} &= \frac{1+\sum_{i=1}^{k}{\frac{1}{i}}}{1+\displaystyle\max\limits_{1\leq i \leq k}{i \times \frac{1}{i}}}
\\ &= \frac{1+\sum_{i=1}^{k}{\frac{1}{i}}}{2}
\\ &\geq \frac{1+\ln(k+1)}{2}
\end{align*}
where $F_k(A_i^-) = \displaystyle\lim_{x \uparrow A_i} F_k(x)$. This shows the ratio can again grow unboundedly as $k \to \infty$. The counterexamples in this subsection show that there are both discrete and continuous profiles, for which our strategy can not achieve the upper bounds within a constant gap, although such profiles may not be common in practice. A natural future direction is to obtain approximate characterizations of the energy harvesting communication channel under general energy harvesting profiles.

\section*{Acknowledgment}
We would like to thank S. Ulukus for stimulating discussions on the topic and the anonymous reviewers for their detailed comments and suggestions which significantly improved the paper. 

\bibliographystyle{ieeetr}
\bibliography{manuscript}

\section*{Appendix}
\label{sec: App}

\begin{proof}[Proof of Theorem \ref{thm: const_gap_rate}]
Here, we provide the detailed proof to the three steps mentioned in the main context. 

{\it Step One:}
Since we are in the case $B_{max} \le E$, each time a non-zero energy packet arrives, the battery will be filled up completely regardless of how much energy was remaining in the battery. In particular, at least $E - B_{max}$ amount of energy is wasted in every non-zero incoming energy packet. Therefore, a system with Bernoulli arrival process $\tilde{E}_t$ where the energy packet size is $\tilde{E} = E - (E- B_{max}) = B_{max}$ is completely equivalent to our original system in terms of the amount of the available energy for communication.

Now, let $g(t)$ be any policy in $\mathcal{G}$. We have: 
\begin{align*}
\liminf_{N \to \infty} \mathbb{E} & \left[ \frac{1}{N}  \displaystyle\sum\limits_{t = 1}^N \frac{1}{2} \log\left( 1 + g( t)\right) \right]
\\
& \stackrel{(a)}{\leq} \liminf_{N \to \infty} \mathbb{E} \left[  \frac{1}{2} \log\left( 1 + \frac{1}{N} \displaystyle\sum\limits_{t = 1}^N  g( t)\right) \right] \\
& \stackrel{(b)}{\leq} \liminf_{N \to \infty} \mathbb{E} \left[  \frac{1}{2} \log\left( 1 +   \frac{1}{N} \displaystyle\sum\limits_{t = 1}^N  \tilde{E}_{t} \right) \right] \\
& \stackrel{(c)}{=} \mathbb{E} \left[  \liminf_{N \to \infty} \frac{1}{2} \log\left( 1 +  \frac{1}{N} \displaystyle\sum\limits_{t = 1}^N  \tilde{E}_{t}\right) \right] \\
& = \; \frac{1}{2} \log\left( 1 + pB_{max}\right)
\end{align*} 
where (a) follows from the concavity of the log which allows to apply Jensen's Inequality; (b) follows from the argument at the beginning that we can equivalently consider a Bernoulli arrival process $\tilde{E}_t$ where the energy packet is $\tilde{E} = B_{max}$  and the fact that $g(t)$ is feasible, so we can not spend more energy up to time $N$ than the amount of exogenous energy we receive by time $N$, i.e.
$$
\sum\limits_{t = 1}^N  g( t)\leq \sum\limits_{t = 1}^N  \tilde{E}_{t};
$$
(c) follows from the Dominated Convergence Theorem. The Dominated Convergence Theorem holds  because first,
\begin{align*}
 \lim_{N \to \infty} \frac{1}{2} &\log\left( 1 +  \frac{1}{N} \displaystyle\sum\limits_{t = 1}^N  \tilde{E}_{t}\right)\\ &\stackrel{(d)}{=}  \frac{1}{2} \log\left( 1 +  \liminf_{N \to \infty} \frac{1}{N} \displaystyle\sum\limits_{t = 1}^N  \tilde{E}_{t} \right) \\
&  \stackrel{(e)}{=} \frac{1}{2} \log\left( 1 +  \mathbb{E}[\tilde{E}_{t}] \right) \\
& = \; \frac{1}{2} \log\left( 1 + pB_{max}\right)
\end{align*} 
almost surely, where (d) follows from the fact $\frac{1}{2}\log(1+x)$ is smooth and monotonically increasing so we can move the $\lim \inf$ inside;
and (e) follows from the Law of Large Numbers; and second 
$$
\frac{1}{2} \log \left( 1 + \frac{1}{N} \displaystyle\sum\limits_{t=1}^{N} \tilde{E}_{t} \right) \le \frac{1}{2} \log \left(1 + B_{max} \right) 
$$
for all $N$ which is finite. Since the above upper bound applies to all feasible policies $g \in \mathcal{G}$, we have the desired upper bound \eqref{eq: policy_upper}.

{\it Step Two:}
Let $\{T_i \}_{i = 1}^L$ be the inter-arrival times between the $i$'th and $i+1$'th non-zero energy packets, where $L$ is the total number of non-zero packets received by time instance N, i.e. $\sum_{i=1}^L T_i \le N < \sum_{i=1}^{L+1} T_i$. Notice that for fixed $N$,  $T_i$'s and $L$ defined in this way are random variables which are functions of $N$ and the random energy arrival process $\{E_t\}_{t=0}^N$. We can lower bound the rate achieved by $g'(t)$ in terms of these new variables as

\begin{align}
 &\liminf_{N \to \infty} \; \mathbb{E}  \left[ \frac{1}{N} \displaystyle\sum\limits_{t = 1}^N \frac{1}{2} \log\left( 1 + g^\prime(t)\right) \right] \nonumber\\
& \;\;\geq \; \liminf_{N \to \infty}  \mathbb{E} \left[   \displaystyle\sum\limits_{i = 1}^L \displaystyle\sum\limits_{j = 0}^{T_i-1} \frac{1}{2} \log\left( 1 + \tilde{g}(j)\right) \bigg / \displaystyle\sum\limits_{i = 1}^{L+1} T_i \right]\label{eq:lln1}
\end{align}
which follows from the fact that the strategy $g^\prime(t)$ we consider is of the form $g^\prime(t) = \tilde{g}(j)$ where $j = t - \max\{t'\le t: E_{t'} = E \}$, i.e., the strategy is invariant across different epochs and the allocated energy depends only on the number of time steps since the last energy arrival.

Notice that as $N \to \infty$, $L \to \infty$ with probability 1. Divide both the numerator and the denominator of the last equation by $L$ and apply the Law of Large Numbers to both the numerator and the denominator, we obtain 
\begin{align}
 &   \liminf_{N \to \infty} \;\displaystyle\sum\limits_{i = 1}^L \displaystyle\sum\limits_{j = 0}^{T_i-1} \frac{1}{2} \log\left( 1 + \tilde{g}(j)\right) \bigg / \displaystyle\sum\limits_{i = 1}^{L+1} T_i \nonumber \\
&\;\; = \;\; \mathbb{E}\left[ \displaystyle\sum\limits_{j = 0}^{T_1-1} \frac{1}{2} \log\left( 1 + \tilde{g}(j)\right)\right]  \bigg / \mathbb{E}[T_1].\label{eq:lln0}
\end{align} 
Note that  $\{ T_i \}$'s are i.i.d. Geometric($p$) so the Law of Large Numbers is directly applicable to the denominator and it  also applies to the numerator since the random variables $\sum\limits_{j = 0}^{T_i-1} \frac{1}{2} \log\left( 1 + \tilde{g}(j)\right)$ are i.i.d. with finite mean.  Now, because the sequence of random variables in \eqref{eq:lln0} converges almost surely, and for every $N$ the sequence is upper bounded, we can apply the Dominated Convergence Theorem to exchange the limit and the expectation in \eqref{eq:lln1} to obtain
\begin{align*}
 \liminf_{N \to \infty}\;  &\mathbb{E}  \left[ \frac{1}{N} \displaystyle\sum\limits_{t = 1}^N \frac{1}{2} \log\left( 1 + g^\prime(t)\right) \right] \\
& \geq \mathbb{E}\left[ \displaystyle\sum\limits_{j = 0}^{T_1-1} \frac{1}{2} \log\left( 1 + \tilde{g}(j)\right)\right]  \bigg / \mathbb{E}[T_1].
\end{align*}

We have 
\begin{align*}
\mathbb{E}&\left[ \displaystyle\sum\limits_{j = 0}^{T_1-1} \frac{1}{2} \log\left( 1 + \tilde{g}(j)\right)\right]  \bigg / \mathbb{E}[T_1]\;\;\;\\
 =  \;\; & p \displaystyle\sum\limits_{i = 1}^{\infty} \mathbb{P}(T_1 = i) \displaystyle\sum\limits_{j = 0}^{i-1} \frac{1}{2} \log(1+\tilde{g}(j)) \\
\stackrel{(a)}{=} \;\; & p \displaystyle\sum\limits_{i = 1}^{\infty} (1-p)^{i-1} p \displaystyle\sum\limits_{j = 0}^{i-1} \frac{1}{2} \log(1+\tilde{g}(j)) \\
  \stackrel{(b)}{=} \;\; &  p \displaystyle\sum\limits_{j = 0}^{\infty} \left( \displaystyle\sum\limits_{i = j+1}^{\infty}  (1-p)^{i-1} p \right) \frac{1}{2} \log(1+\tilde{g}(j)) \\
   \stackrel{(c)}{=} \;\;  & \displaystyle\sum\limits_{j = 0}^{\infty} p(1-p)^j \frac{1}{2} \log(1 + \tilde{g}(j)) \\
   = \;\; &  \displaystyle\sum\limits_{j = 0}^{\infty} p(1-p)^j \frac{1}{2} \log(1 + p(1-p)^j B_{max}) 
\end{align*} 
where (a) follows from the fact that $\{ T_1 \}$ is Geometric($p$), (b) follows from switching the order of summations, and (c) uses the formula for the sum of geometric series.  

{\it Step Three:}
Finally, to complete the proof, we need to show:
\begin{align}
\label{eq: constant_gap_upper}
\frac{1}{2}  \log(1 + p B_{max}) - \displaystyle\sum\limits_{j = 0}^{\infty}  p (1-p)^j  \frac{1}{2} \log(1 + \tilde{g}(j)) \le 0.973
\end{align}

In order to show \eqref{eq: constant_gap_upper}, we will first restrict the range of $(B_{max}, p)$ to be considered by noticing:
\begin{align*}
\frac{1}{2}  \log(1 + p B_{max}) - & \displaystyle\sum\limits_{j = 0}^{\infty}  p (1-p)^j  \frac{1}{2} \log(1 + \tilde{g}(j)) \\
& \le \frac{1}{2}  \log(1 + p B_{max}) \\
& \le 0.973
\end{align*}
for all $(B_{max}, p)$ such that $pB_{max} \le 2.853$. 

So let's restrict ourself to the set of $(B_{max}, p)$ such that $pB_{max} > 2.853$. We have
\begin{align*}
& \frac{1}{2}  \log(1 + p B_{max}) - \displaystyle\sum\limits_{j = 0}^{\infty}  p (1-p)^j  \frac{1}{2} \log(1 + \tilde{g}(j)) \\
= \; & \frac{1}{2}  \log(1 + p B_{max}) \\
& - \displaystyle\sum\limits_{j = 0}^{\infty}  p (1-p)^j  \frac{1}{2} \log(1 + p(1-p)^j B_{max}) \\
\stackrel{(a)}{\le} \; & \frac{1}{2} \log(1 + p B_{max}) - \displaystyle\sum\limits_{j = 0}^{\infty}  p (1-p)^j  \frac{1}{2} \log(p(1-p)^j B_{max})) \\
 = \; & \frac{1}{2}\log(p) + \frac{1}{2}\log(B_{max}) + \frac{1}{2}\log\left(\frac{1}{pB_{max}} + 1\right) \\
& - \displaystyle\sum\limits_{j = 0}^{\infty} p(1-p)^j \frac{1}{2} \left[\log(p) +  j\log\left(1-p\right) +\log( B_{max})\right] \\
\stackrel{(b)}{=} \; & \frac{1}{2}\log(p) + \frac{1}{2}\log(B_{max}) +  \frac{1}{2}\log\left(\frac{1}{pB_{max}} + 1\right) \\
& - \frac{1}{2}\log(p) -  \frac{1-p}{p} \frac{1}{2} \log\left(1-p\right)  -\frac{1}{2} \log( B_{max}) \\
= \; &   \frac{1}{2}\log\left(\frac{1}{pB_{max}} + 1\right) -  \frac{1-p}{p} \frac{1}{2} \log\left(1-p\right) \\
\stackrel{(c)}{\le}  \; & \frac{1}{2\ln(2)}\frac{1}{pB_{max}} -  \frac{1-p}{p} \frac{1}{2} \log\left(1-p\right) \\
\stackrel{(d)}{\le}  \; & \frac{1}{2\ln(2)}\frac{1}{2.853} +  \frac{1-p}{2p} \log\left(\frac{1}{1-p}\right) \\
\stackrel{(e)}{\le} \; & 0.253 + 0.72 \\
= \; & 0.973
\end{align*}
where (a) follows from the fact that removing the $1$ inside the second $\log$ results in an upper bound; (b) uses the identity $\displaystyle\sum\limits_{j=0}^{\infty} j \cdot p(1-p)^j  = (1-p) \mathbb{E}\left[X\right] =  \frac{1-p}{p}$, where $X \sim$  Geometric($p$), and $\displaystyle\sum\limits_{j = 0}^{\infty} p(1-p)^j = 1$; (c) uses the inequality $\ln(1 + x) \le x$ for all $x$; (d) follows from the fact we restrict ourself to the case $pB_{max} > 2.853$; and  finally (e) follows from the fact that
\begin{align*}
g(p) & \triangleq \frac{1-p}{2p} \log\left(\frac{1}{1-p}\right)
\end{align*}
is a continuous bounded function on $p \in (0,1)$. Furthermore, it's monotonically decreasing and is upper bounded by $\displaystyle\lim_{p \to 0} g(p) = \frac{1}{2\ln(2)} = 0.72$. 

\end{proof}

\begin{proof}[Proof of Theorem \ref{thm: const_gap_rate2}]
The highlight of the proof has been outlined in the comment after the statement of Theorem \ref{thm: const_gap_rate2} in Section \ref{subsec: rate_B_greater_E}. The proof is essentially the same as Theorem \ref{thm: const_gap_rate} with $B_{max}$ replaced by E. Therefore, we will simply reiterate the major steps of the proof without boring the reader with all the details.

{\it Step One:} Using Jensen's Inequality, Dominated Convergence Theorem and Law of Large Numbers, we have the following upper bound: 
\begin{align*}
\max_{g \in \mathcal{G}} \liminf_{N \to \infty} \mathbb{E} \left[ \frac{1}{N} \displaystyle\sum\limits_{t = 1}^N \frac{1}{2} \log\left( 1 + g(t)\right) \right] \leq \frac{1}{2}\log(1+pE)
\end{align*}

%
%

{\it Step Two:}
Replace $B_{max}$ by $E$ in Step Two in the proof of Theorem \ref{thm: const_gap_rate} and follow the exact same sequence of arguments, we can lower bound the rate achieved by $g^\prime(t)$ as:
\begin{align*}
\liminf_{N \to \infty}  & \; \mathbb{E} \left[ \frac{1}{N} \displaystyle\sum\limits_{t = 1}^N \frac{1}{2} \log\left( 1 + g^\prime(t)\right) \right] \\
& \ge  \displaystyle\sum\limits_{j = 0}^{\infty} p(1-p)^j \frac{1}{2} \log(1 + p(1-p)^j E) 
\end{align*}

{\it Step Three:}
Finally, replace $B_{max}$ everywhere by $E$ in Step Three in the proof of Theorem \ref{thm: const_gap_rate}, we have the following bound:
\begin{align*}
\frac{1}{2}  \log(1 + p E) - \displaystyle\sum\limits_{j = 0}^{\infty}  p (1-p)^j  \frac{1}{2} \log(1 +  p(1-p)^j & E ) \\
  & \le 0.973
\end{align*}
which completes the proof.

\end{proof}

\begin{proof}[Proof of Theorem \ref{thm: Achievable}]
We want to show that for any $\{ \mathcal{E}_j\}$ satisfying \eqref{eq: EnergyConstr} and any $\epsilon > 0$, we can construct a communication strategy that achieves a rate
\begin{align*}
 R \ge \displaystyle\sum\limits_{j = 0}^{\infty} p & (1-p)^j \max_{p(x): |X|^2 \le \mathcal{E}_j} I(X; Y) - \epsilon
\end{align*} 
and has arbitrarily small probability of error.

We start by fixing a positive integer $M(\epsilon)$, and positive real numbers $\delta_2(\epsilon, M)$ and $\delta_1(\epsilon, M, \delta_2)$ such that
\begin{align}
& (1 - p)^{M+1} K \le \epsilon/3  \label{eq: M_Constraint} \\
\frac{\delta_2}{1+\delta_2} & \left( 1 - (1-p)^{M+1} \right) K \le \epsilon/3 \label{eq: delta2_Constraint} \\
& \delta_1 \cdot \frac{(M+1)p}{1 + \delta_2} K \le \epsilon/3 \label{eq: delta1_Constraint}
\end{align}
where $K \triangleq \displaystyle\max_{j} \left\{  \displaystyle\max_{p(x): |X|^2 \le \mathcal{E}_j} I(X; Y) \right\}$ is a finite constant. Furthermore, let $L$ be a large positive integer and let
\begin{align*}
n^{(j)} = L\left( (1-p)^j - \delta_1  \right)
\end{align*} 

For a given energy allocation policy $\{ \mathcal{E}_j\}$ for $j = 0,1,2,...$ that satisfies \eqref{eq: EnergyConstr}, the transmitter and the receiver agree on a sequence of $M+1$ codebooks:
\begin{align*}
\mathcal{C}^{(0)}, \mathcal{C}^{(1)}, \mathcal{C}^{(2)}, ..., \mathcal{C}^{(M)}
\end{align*}
where each codebook $\mathcal{C}^{(j)}$ consists of $2^{n^{(j)} R^{(j)}}$ codewords generated randomly from a distribution $p(x)$ whose support is limited to $[-\sqrt{\mathcal{E}_j},\sqrt{\mathcal{E}_j}]$ so that the symbols of each codeword are amplitude constrained according to 
\begin{align*}
|X_i^{(j)}(k)|^2 \le \mathcal{E}_j
\end{align*}
for all $i = 1,2, ..., n^{(j)}$, $j = 0,1,2..., M$ and $k = 1,2,..., 2^{n^{(j)} R^{(j)}}$.

During the course of communication the transmitter will aim to send one codeword from each of the $M+1$ codebooks. Communication proceeds as follows: At each time step $t$, the transmitter sees the realization of the energy process $E_t$, let $j = t - \max\{t'\le t: E_{t'} = E \}$, i.e. the number of time steps since the last time battery was recharged. Then the transmitter  selects the next untransmitted symbol from the codeword it wants to transmit from the codebook $\mathcal{C}^{(j)}$ (if all $n^{(j)}$ symbols have been transmitted, it simply transmits the zero symbol). Similarly if $j>M$, it transmits the zero symbol. Communication ends when the transmitter observes the arrival of the $L+1$'th energy packet. (We assume that communication starts with the arrival of the first energy packet). The receiver can track the codebook used by the transmitter and decode each codeword separately by collecting together the corresponding channel observations since we assume that $\{ E_t \}$ are also causally known at the receiver.

We will say that communication is in error if one of the following events occur:
\begin{itemize}
\item $E^{(j)}, j=0,\dots,M:$ by the end of the transmission, the codeword from $\mathcal{C}^{(j)}$ has not been completely transmitted.\looseness=-10
\item $E_0:$ the total duration of the communication exceeds $\frac{L}{p}(1 + \delta_2)$.\looseness=-10
\item $F^{(j)}, j=0,\dots,M:$ the codeword from codeword $\mathcal{C}^{(j)}$ is decoded erroneously at the receiver.\looseness=-10
\end{itemize}
Below we will argue that the probability of each one of these events can be made arbitrarily small by taking $L$ large enough.

To bound the probability of $E^{(j)}$, recall the inter arrival time of energy packets, denoted by  $T_1, T_2, ..., T_L$, are i.i.d. Geometric($p$) r.v's. Therefore,
\begin{align*}
\mathbb{P}(E^{(j)}) & =  \mathbb{P}\left(  \displaystyle\sum\limits_{i = 1}^{L} \mathbbm{1}\{ T_i >j \} \le n^{(j)}\right) \\
& =  \mathbb{P}\left(  \frac{1}{L} \displaystyle\sum\limits_{i = 1}^{L} \mathbbm{1}\{ T_i >j \} \le  (1-p)^j - \delta_1  \right) \\
 & \le  \epsilon_1 ( L) \rightarrow 0
\end{align*}
as $L \to \infty$ by the weak law of large number, since $\mathbb{E}[ \mathbbm{1}\{ T_i >j \} ] = \mathbb{P}(  T_i >j) = (1-p)^j$. 

Furthermore, notice the total duration of communication $N = T_1 + \cdots + T_L$ is a random quantity. Then, we have
\begin{align*}
\mathbb{P}(E_0) & =  \mathbb{P}\left(  \displaystyle\sum\limits_{i = 1}^{L} T_i > \frac{L}{p}(1 + \delta_2) \right) \\
& =   \mathbb{P}\left(  \frac{1}{L} \displaystyle\sum\limits_{i = 1}^{L} T_i > \frac{1}{p}(1 + \delta_2) \right) \\
 & \le  \epsilon_2 ( L) \rightarrow 0
\end{align*}
as  $L \to \infty$, again by the weak law of large numbers, since $\mathbb{E}[ T_i ] = 1/p$. 

On the receiver side, the decoder can decode each individual codeword via joint typical decoding once the transmission is finished. The classic channel coding theorem shows that  as long as 
$$R^{(j)} \leq \displaystyle\max_{p(x): |X|^2 \le \mathcal{E}_j} I(X; Y) $$ 
the probability of each one of the events $F^{(0)},\dots, F^{(M)}$ can be made arbitrarily small by making $L$, hence $n^{(j)}$, large enough, using codewords from each codebook $\mathcal{C}^{(j)}$.

Finally, the union bound allows us to conclude that the probability of the union of all these error events can be made arbitrarily small as $L \to \infty$ (note that $\delta_1, \delta_2 > 0$, and $M$ are fixed at the beginning).

We finally compute the average rate achieved by this communication strategy
\begin{align*}
& R \ge  \frac{p}{L(1+ \delta_2)} \displaystyle\sum\limits_{j = 0}^{M} n^{(j)} R^{(j)} \\
 = &  \frac{p}{L(1+ \delta_2)} \displaystyle\sum\limits_{j = 0}^{M} L\left( (1-p)^j - \delta_1  \right)\max_{p(x): |X|^2 \le \mathcal{E}_j} I(X; Y)   \\
 = & \displaystyle\sum\limits_{j = 0}^{M} \frac{p}{1+\delta_2} \left( (1-p)^j - \delta_1  \right) \max_{p(x): |X|^2 \le \mathcal{E}_j} I(X; Y)  \\
=   & \displaystyle\sum\limits_{j = 0}^{\infty} p(1-p)^j \max_{p(x): |X|^2 \le \mathcal{E}_j} I(X; Y)   \\ 
& - \displaystyle\sum\limits_{j = M+1}^{\infty} p(1-p)^j \max_{p(x): |X|^2 \le \mathcal{E}_j} I(X; Y)   \\ 
& -  \displaystyle\sum\limits_{j = 0}^{M} \frac{\delta_2 p}{1+\delta_2} (1-p)^j \max_{p(x): |X|^2 \le \mathcal{E}_j} I(X; Y) \\
 & - \displaystyle\sum\limits_{j = 0}^{M} \frac{p}{1+\delta_2} \delta_1 \max_{p(x): |X|^2 \le \mathcal{E}_j} I(X; Y) \\
\ge  & \displaystyle\sum\limits_{j = 0}^{\infty} p(1-p)^j \max_{p(x): |X|^2 \le \mathcal{E}_j} I(X; Y) - (\epsilon/3 +\epsilon/3+\epsilon/3) 
\end{align*}
which is what we set up to prove. The last step follows from the fact $M$, $\delta_2$ and $\delta_1$ are chosen to satisfy the bounds \eqref{eq: M_Constraint}, \eqref{eq: delta2_Constraint} and \eqref{eq: delta1_Constraint}.

\end{proof}
\begin{proof}[Proof of Theorem~\ref{thm:CapCSIR}]\footnote{This proof was pointed out to the authors by one of the anonymous reviewers.} 
 The capacity of a general channel is given by \cite[Theorem 3.2.1]{Han} as the maximum \emph{spectral inf-mutual information rate} between the input and the output,
 $$
 C= \sup_X\text{p-} \liminf_{n\rightarrow\infty} \frac{1}{n}\log \frac{P_{Y^n\vert X^n}({y^n\vert x^n})}{P_{Y^n}(y^n)},
 $$
where the supremum is over all input processes $X=\{X^n\}_{n=1}^\infty$. When the receiver has access to a side information process $G=\{G^n\}_{n=1}^\infty$, this process can be viewed as part of the output of the channel and the corresponding capacity becomes 
$$
C_G =\sup_X\; \text{p-}\liminf_{n\rightarrow\infty}\: \frac{1}{n}\log \frac{P_{Y^n,G^n\vert X^n}(y^n,g^n \vert x^n)}{P_{Y^n,G^n}(y^n,g^n)}. 
$$
We have
\begin{align}
C_G &= \sup_X\; \text{p-}\liminf_{n\rightarrow\infty}\: \frac{1}{n}\log \frac{P_{Y^n,G^n\vert X^n}(y^n,g^n \vert x^n)}{P_{Y^n,G^n}(y^n,g^n)} \nonumber\quad\\
&= \sup_X\; \text{p-}\liminf_{n\rightarrow\infty}\: \Big[ \frac{1}{n}\log \frac{P_{Y^n\vert X^n}(y^n \vert x^n)}{P_{Y^n}(y^n)}\nonumber \\
&\qquad\qquad + \frac{1}{n}\log \frac{P_{G^n\vert Y^n,X^n}(g^n \vert y^n , x^n)}{P_{G^n\vert Y^n}(g^n\vert y^n)} \Big] \nonumber
\\ 
& \le \sup_X  \text{p-}\liminf_{n\rightarrow\infty}\:  \frac{1}{n}\log \frac{P_{Y^n\vert X^n}(y^n \vert x^n)}{P_{Y^n}(y^n)}   \nonumber\\
&\quad + \sup_X \text{p-}\limsup_{n\rightarrow\infty}  \frac{1}{n}\log \frac{P_{G^n\vert Y^n,X^n}(g^n \vert y^n , x^n)}{P_{G^n\vert Y^n}(g^n\vert y^n)} \label{sum2sup} \\
& = \: C \: + \sup_X \text{p-}\limsup_{n\rightarrow\infty} \frac{1}{n}\log \frac{P_{G^n\vert Y^n,X^n}(g^n \vert y^n , x^n)}{P_{G^n\vert Y^n}(g^n\vert y^n)}  \nonumber\\
& \le \: C \: + \sup_X \text{p-}\limsup_{n\rightarrow\infty}  \frac{1}{n}\log \frac{1}{P_{G^n\vert Y^n}(g^n\vert y^n)} \label{Probless1}\\
& = C \: + \sup_{X} \bar{H}(G\vert Y) \label{VH:CondEntDef}\\
& \le C \: + \sup_{X} \bar{H}(G). \label{VH:I}
\end{align}
Here, (\ref{sum2sup}) follows from the fact that for any two sequences of random variables $(V_n)_{n=1}^{\infty}$ and $(Z_n)_{n=1}^{\infty}$ we have \cite[p.15]{Han}
\begin{equation*}
\text{p-}\liminf_{n\rightarrow \infty}(V_n + Z_n) \le \text{p-}\liminf_{n\rightarrow \infty}V_n 
\: + \: \text{p-}\limsup_{n\rightarrow \infty}Z_n.
\end{equation*} 
Note that $P_{G^n\vert Y^n,X^n}(g^n \vert y^n , x^n) \le 1$, thus (\ref{Probless1}) holds. (\ref{VH:CondEntDef}) follows by the definition of the \textit{conditional spectral sup-entropy rate} \cite[p.182]{Han}, given as
\begin{equation*}
\bar{H}(G\vert Y)=\text{p-}\limsup_{n\rightarrow\infty}\frac{1}{n}\log \frac{1}{P_{G^n \vert Y^n}(g^n\vert y^n)}.
\end{equation*}
Finally, non-negativity of the inf-mutual information rate \cite[Equations 3.2.3, 3.2.10]{Han}
\begin{equation*}
0 \le \underline{I}(G;Y) \le \bar{H}(G) - \bar{H}(G\vert Y), 
\end{equation*}
yields (\ref{VH:I}). In (\ref{VH:I}), we assume that the side information process $G$ is independent of the input process, which yields the conclusion of the theorem. Note that when the side information process is a function of the input, the capacity gain due to receiver side information can be bounded by considering the sup of the spectral sup-entropy rate of $G$ over the input process. 
\end{proof}
\begin{proof}[Proof of Theorem \ref{thm: Achievable2}] 
Here, we provide an alternative proof for Theorem \ref{thm: Achievable2} that does not use Theorem~\ref{thm:CapCSIR}. We first define a specific finite state channel and establish a relation between its capacities with and without channel state information at the receiver, denoted by $C_1$ and $C_2$ respectively. In particular, we show that $C_2\geq C_1-H(p)$. We then relate the capacity of this finite state channel to our original communication setup defined in Section~\ref{sec: SysModel}. We show that the capacity $C$ of our original energy harvesting channel without channel state information at the receiver satisfies $C\geq C_2$, while $C_1$ is larger than the right-hand side of \eqref{eq: Achievable} in Theorem \ref{thm: Achievable}. This completes the proof of the theorem.

Given any $\mathcal{E}_0,\mathcal{E}_1,\mathcal{E}_2,\dots$ that satisfy \eqref{eq: EnergyConstr2} define a finite state channel $p(y_t|x_t, s_t)$  where the output $Y_t\in\mathbb{R}$ depends on the input $X_t\in\mathbb{R}$ to the channel and the channel state $S_t\in\{0,1,..., M+1\}$ in the following way:
\begin{align}
\label{eq:fsc}
Y_t = \left\{\begin{array}{rl}
X_t + Z_t & \mbox{if } \;\;\; |X_t|^2 \le \mathcal{E}^\prime_{S_t}\\
Z_t & \mbox{otherwise}
\end{array}  \right.
\end{align}
where $Z_t$ is i.i.d. $\mathcal{N}(0,1)$ Gaussian noise. The channel state process $\{ S_t \}$ is a Markov Chain and its $(M+2) \times (M+2)$ transition matrix is given by
\begin{equation}
\label{eq: Tm}
T=\left( \begin{array}{ccccc}
p & 1-p & 0 & \cdots & 0 \\
p & 0 & 1- p &  \ddots & \vdots  \\
\vdots & \vdots & \ddots & \ddots   &  0\\
p & 0 &  & 0    &  1-p\\
p & 0 & \cdots &   0  & 1-p
 \end{array} \right)
 \end{equation}
where $T_{ij}=p(S_{t+1} = j|S_{t} = i)$. The sequence $\mathcal{E}^\prime_0, \dots, \mathcal{E}^\prime_{M+1}$ satisfies $\mathcal{E}^\prime_j=\mathcal{E}_j$ for $0\leq j\leq M$ and $\mathcal{E}^\prime_{M+1}=0$. There are no constraints on the transmitted signal $X_t\in\mathbb{R}$ over this channel. We assume that the initial state of the channel $S_0=0$ and the realization of the state sequence $\{S_t\}$ is known causally at the transmitter. We will discuss the capacity of this channel under two different assumptions: 
\begin{enumerate}
\item{The realization of the state sequence $\{S_t\}$ is also known at the receiver. In this case we will denote the capacity by $C_1$.}
\item{The realization of the state sequence $\{S_t\}$ is not  known at the receiver. In this case we will denote the capacity by $C_2$.}
\end{enumerate}
Obviously, $C_2\leq C_1$. Below, we show that $C_2$ can not be much smaller than $C_1$.

First note that the Markov Chain $\{S_t\}$ is time-invariant, aperiodic and indecomposable. For such channels Theorem 2 of \cite{DasNarayan} \footnote {Theorem 2 of \cite{DasNarayan} is stated for multiple access finite state channel. We use a simplified version of it for single-user channels, combined with part 2 of Corollary 1 in the same paper which shows that when the initial state is known to the transmitter, the capacity region is independent of the initial pmf.} establishes the following expression for the capacity,

\begin{align}
\label{eq: cap_rec_info}
C_D = \sup_{\underline{F}} \liminf_{n \to \infty}  \frac{1}{n} I (F^n ; Y^n |D^n)
\end{align}
where $D^n=(D_1, ..., D_n)$ is the channel state information available at the receiver which is of the form $D_t=g(S_t)$ for a given mapping $g(\cdot)$. Here, we will be interested in the two extremal cases when $D^n=S^n$, in which case the capacity is denoted by $C_1$, and $D^n=\emptyset$ in which case the capacity is denoted by $C_2$. Here, $F^n$ denotes $(F_{1}, \ldots , F_{n})$ where $F_{t}$ stands for a random variable that takes values in the space of all mappings $\lbrace f_t : E^t \to \mathcal{X} \rbrace$. Also, $\underline{F}$ denotes a sequence $\left\{ F^n \right\}_{n=1}^{\infty}$. Now, observe that
\begin{align}
C_1 & = \sup_{\underline{F}} \liminf_{n \to \infty}  \frac{1}{n} I (F^n ; Y^n |S^n)\nonumber \\
& \le \sup_{\underline{F}} \liminf_{n \to \infty}  \frac{1}{n} I (F^n ; Y^n , S^n)\nonumber \\
& =  \sup_{\underline{F}} \liminf_{n \to \infty} \frac{1}{n} \left( I (F^n ; Y^n) + I(F^n;S^n|Y^n) \right) \nonumber\\
& \le \sup_{\underline{F}} \liminf_{n \to \infty} \frac{1}{n}  I (F^n ; Y^n) + \frac{1}{n}H(S^n) \label{eq:eq1}\\
& =  H(p) +  \sup_{\underline{F}} \liminf_{n \to \infty} \frac{1}{n}  I (F^n ; Y^n) \label{eq:eq2}\\
& = H(p)+ C_{2} \label{eq:CSIRCapDiff}
\end{align}
where \eqref{eq:eq2} follows from the fact that 
 \begin{align*}
\lim_{n\to \infty}\frac{1}{n}H(S^n) & = \lim_{n \to \infty}{H(S_{n+1} | S_{n},\ldots S_{1})}\\
& = \lim_{n \to \infty}{H(S_{n+1} | S_{n})}\\
 & = H(p)
\end{align*}
since for the transition matrix $T$ in \eqref{eq: Tm}, we have ${H(S_{n+1} | S_{n}})=H(p)$ independent of the state $S_n$. $H(p)$ refers to the binary entropy function evaluated at $p$. Note that when $\lim_{n\to \infty}\frac{1}{n}H(S^n)$ exists,  we can decompose the limit inferior of the sum of the two sequences in \eqref{eq:eq1} to the sum of the limit inferior of those sequences.  

This shows that the capacity of this channel without side information at the receiver is at most $H(p)$ bits smaller than its capacity with receiver side information. We next relate the capacity of this finite state channel to the capacity of our original channel. 

In the proof of Theorem~\ref{thm: Achievable}, we have shown that when the realization of the external energy arrival process $\{E_t\}$ is causally known to both the transmitter and the receiver, for any given energy allocation policy $\{ \mathcal{E}_j\}$ that satisfies \eqref{eq: EnergyConstr}, we can achieve a rate 
\begin{align*}
 R(M) \ge \displaystyle\sum\limits_{j = 0}^{\infty} p & (1-p)^j \max_{p(x): |X|^2 \le \mathcal{E}_j} I(X; Y) - \epsilon(M),
\end{align*} 
where $M$ determines  the maximal number of channel uses employed after each energy packet arrival and as $M\to\infty$, $\epsilon(M)\to 0$. Observe that the same strategy can be employed over the finite state channel defined in \eqref{eq:fsc} and would achieve exactly the same rate $R(M)$. (Here, the state $S_t$ will dictate which of the $M+1$ codebooks $\mathcal{C}^{(0)}, \mathcal{C}^{(1)}, \mathcal{C}^{(2)}, ..., \mathcal{C}^{(M)}$ to transmit from at given $t$.) This is because the code constructed in the achievability scheme for our original channel (with receiver side information) will satisfy $|X_t|^2 \le \mathcal{E}_{S_t}$ at every $t$  when applied over the finite state channel. This ensures that the input-output relations for the two channels are the same. Moreover, the transition matrix $T$ in \eqref{eq: Tm} of the finite state channel is designed so that the process $\{\mathcal{E}^\prime_{t}\}$  induced by $\{S_t\}$ is probabilistically equivalent to the energy allocation process $\{\mathcal{E}_{t}\}$ induced by the Bernoulli energy arrival process $\{E_{t}\}$ over the original channel. This allow us to conclude that 
$$
C_1\geq R(M),
$$
or using the result of \eqref{eq:CSIRCapDiff}
$$
C_2\geq R(M)-H(p).
$$
Finally, we want to argue that the capacity of our original system \underline{without} receiver side information is lower bounded by the capacity of this finite state channel \underline{without} receiver side information, i.e. $C \ge C_2$. Once this is established, we can conclude that $C \ge R(M) - H(p)$. Taking $M\to\infty$, hence $\epsilon(M)\to 0$, completes the proof of the theorem. 

To show that $C \ge C_2$, we argue that any code designed for the finite state channel (without receiver state information) can be applied to our original energy harvesting communication channel (again without receiver side information).

When the state information $\{S_t\}$ (or equivalently the realization of the process $\{\mathcal{E}^\prime_t\}$) is known causally at the transmitter, a code  for communicating over the finite state channel is given by a set of mappings
$$
W \to f_W^n=(f_{W, 1}, f_{W,2},..., f_{W,n})  
$$
for each message $W$ (a so-called Shannon strategy), where  for given $W$, $f_{W,t}$ is a mapping from all the transmitter side information received up to time step $t$, i.e. $S^t = (S_0, ..., S_t)$, to the space of channel input symbols, $\mathcal{X}$. This code can not be immediately applied to our original energy harvesting communication channel since the input alphabet for the finite state channel $\mathcal{X}=\mathbb{R}$ and it can potentially imply using symbols with magnitude $|X_t|^2 > \mathcal{E}^\prime_{S_t}$. However, this code can be modified, without impacting its probability of error and rate, to a form that can be immediately applied over our energy harvesting channel. 

Let us modify the code by zeroing all the symbols that violate $|X_t|^2 > \mathcal{E}^\prime_{S_t}$. In particular, we define a new sequence of functions $(\overline{f}_{W,1}, \overline{f}_{W,2}, ..., \overline{f}_{W,n})$ associated with the message $W$ as follow: 
$$
\overline{f}_{W,t}(S^t) = \mathbbm{1}_{\left\{ |f_{W,t}(S^t)|^2 \le \mathcal{E}^\prime_{S_t} \right\}} f_{W,t}(S^t).
$$

Based on the way our finite state channel is constructed, the output distributions are exactly the same whether we use the code with $f^n$ or $\overline{f}^n$, i.e. 
\begin{align*}
p(y^n|f_W^n, s^n) = p(y^n|\overline{f}_W^n, s^n) 
\end{align*}
Therefore, if we keep our decision regions at the decoder exactly the same, we can achieve the exact same rate and probability of error with both codes over the finite state channel. Now that all the potential channel input symbol generated from this new code do satisfy the constraint $|X_t|^2\leq\mathcal{E}^\prime_{S_t}$, we can apply this new code over our energy harvesting communication channel: Given the Bernoulli energy harvesting process $\{E_t\}$ define the process $\{S_t^\prime\}$ such that 
$$
S_t^\prime=\left\{\begin{array}{rl} j &\textrm{if}\quad 0\leq j\leq M,\\
M+1  &\textrm{if}\quad  j>M.
\end{array}  \right.
$$
where $j = t - \max\{t'\le t: E_{t'} = E \}$. Applying the codebook $\{f_W^n\}_W$ with side information process $\{S_t^\prime\}$ over the original energy harvesting communication channel achieves the exact same rate and probability of error as in the finite state channel. This is because the input-output relations for the two channels are the same and $\{S_t^\prime\}$  and $\{S_t\}$ are probabilistically equivalent. Also note that because of our initial assumption that $\mathcal{E}_0,\mathcal{E}_1,\mathcal{E}_2,\dots$ satisfy \eqref{eq: EnergyConstr2}, this is an energy feasible strategy for our harvesting system. This completes the proof of the theorem.
\end{proof}

\begin{proof}[Proof of Theorem \ref{thm: ConstantGapBlessE}]
The proof follows almost the same line of argument as in Theorem \ref{thm: const_gap_rate}, except we have a different constant here due to the additional $K(p) = 1.04 + H(p)$ in the $C_{lb}(B_{max}, p)$.

Again, we will first restrict the range of $(B_{max}, p)$ to be considered by noticing:
\begin{align*}
C_{ub}(B_{max}, p) - & C_{lb}(B_{max}, p )  \le C_{ub}(B_{max}, p) \\
 & = \frac{1}{2} \log(1 + pB_{max}) \\
 & \le 2.58
\end{align*}
for all $(B_{max}, p)$ such that $pB_{max} \le 34.75$. 

Next, let's restrict ourself to the set of $(B_{max}, p)$ such that $pB_{max} > 34.75$:
\begin{align*}
& C_{ub}(B_{max}, p) - C_{lb}(B_{max}, p ) \\
= & \frac{1}{2} \log \left(1 + p B_{max}\right) \\
&-  \left( \displaystyle\sum\limits_{j = 0}^{\infty} p(1-p)^j \frac{1}{2} \log\left( 1 + (1-p)^j p B_{max}  \right) - K(p)\right)^+ \\
\stackrel{(a)} \le &  \frac{1}{2} \log \left(1 + p B_{max}\right) \\ 
& - \displaystyle\sum\limits_{j = 0}^{\infty} p(1-p)^j \frac{1}{2} \log\left( 1 + (1-p)^j p B_{max}  \right) + H(p) + 1.04 \\
 \stackrel{(b)} \le & \frac{1}{2\ln(2)}\frac{1}{pB_{max}} -  \frac{1-p}{p} \frac{1}{2} \log\left(1-p\right) + H(p) + 1.04 \\
 \stackrel{(c)} \le & \frac{1}{2\ln(2)}\frac{1}{34.75} -  \frac{1-p}{p} \frac{1}{2} \log\left(1-p\right) + H(p) + 1.04 \\
 = & 1.06 + \frac{1-p}{2p} \log\left(\frac{1}{1-p}\right) + H(p) \\
 \stackrel{(d)} \le & 2.58
\end{align*}
where (a) comes from removing the $(\cdot)^+$ resulting in an upper bound and using the definition of $K(p)$. (b) uses exactly the same sequence of inequalities in the proof of Theorem \ref{thm: const_gap_rate}, so we won't repeat here. (c) comes from the fact we restrict ourself to $pB_{max} > 34.75$. Finally, (d) comes from the fact:
\begin{align*}
g(p) & \triangleq \frac{1-p}{2p} \log\left(\frac{1}{1-p}\right) + H(p) \\
&= \frac{1-p}{2p} \log\left(\frac{1}{1-p}\right) + p \log\left( \frac{1}{p} \right) + (1-p) \log\left( \frac{1}{1-p} \right)
\end{align*}
is a continuous bounded function on $p \in (0,1)$. Furthermore, it's concave and attains a maximum value of 1.52 at  $p = 0.413$. 
\end{proof}

\end{document}